\documentclass[11pt]{article}

\usepackage[a4paper,margin=1in]{geometry}

\usepackage{amsmath,amssymb,amsthm}
\usepackage{mathtools}
\mathtoolsset{showonlyrefs} 

\usepackage{algorithm}
\usepackage{algpseudocode}

\usepackage{graphicx}
\usepackage{subcaption}

\usepackage{enumitem}

\usepackage[colorlinks=true,
            linkcolor=blue,
            citecolor=blue,
            urlcolor=blue]{hyperref}

\usepackage{bm}

\renewcommand{\d}{\mathop{}\!\mathrm{d}}
\newtheorem{lemma}{Lemma}

\newtheorem{corollary}{Corollary}
\newtheorem{theorem}{Theorem}
\newtheorem{remark}{Remark}

\numberwithin{equation}{section}

\theoremstyle{plain}
\newcounter{appthm}

\theoremstyle{remark}

\usepackage{xcolor}

\title{An Analytic COS Method for Compound Option Valuation}
\author{
Zhipeng Huang 
\thanks{Corresponding author. Mathematical Institute, Utrecht University, Utrecht, The Netherlands (\texttt{z.huang1@uu.nl}).}
\and
Cornelis W.~Oosterlee
\thanks{Mathematical Institute, Utrecht University, Utrecht, The Netherlands (\texttt{c.w.oosterlee@uu.nl}).}}

\date{}

\begin{document}

\maketitle

\begin{abstract}
We develop an analytic Fourier cosine (COS) method for the valuation of compound options. By deriving closed-form expressions for the cosine coefficients at all compound stages, the proposed method eliminates the need for numerical quadrature in intermediate exercise stages while retaining the convergence properties of the underlying COS approximation. The formulation extends to multi-stage compound structures and a broader class of payoffs, and remains applicable to a wide class of stochastic models characterized by known characteristic functions, including jump-diffusion dynamics. Numerical experiments demonstrate improved computational efficiency compared with quadrature-based implementations while maintaining high accuracy. Applications to staged real-option problems further illustrate the flexibility of the method in handling nested decision structures under different uncertainty dynamics.
\end{abstract}


\medskip


\tableofcontents 

\section{Introduction}

Compound options are derivatives whose underlying asset is itself an option. They were introduced by Geske \cite{Geske1979}, who derived a closed-form valuation formula under the Black--Scholes model \cite{BlackScholes1973,Merton1973}.

Beyond financial derivatives, compound-option structures arise in multi-stage investment problems such as real options, staged R\&D investments, and sequential project expansion \cite{DixitPindyck1994,Trigeorgis1996}. In such settings, each decision stage grants the right, but not the obligation, to proceed to the next phase by paying an investment cost. The resulting valuation problem involves nested continuation values across successive decision dates.

Mathematically, these problems lead to nested conditional expectations of future project values. While this structure provides a natural framework for capturing managerial flexibility under uncertainty, it also introduces computational challenges due to the recursive nature of the continuation values. These challenges become pronounced when realistic stochastic dynamics beyond the Black--Scholes model are considered, since each stage introduces an additional conditional expectation layer.

To illustrate the structure, consider a two-stage setting with decision times $0<T_1<T$. Let $S(t)$ denote the project value at time $t$, $K_1$ the investment cost in the first stage, and $K$ the terminal exercise threshold. At time $T_1$, the value of continuing the project is given by
\begin{equation}\label{1}
\left( V_{\mathrm{inner}}(T_1, S(T_1); T, K) - K_1 \right)^+,
\end{equation}
where
\begin{equation}\label{2}
V_{\mathrm{inner}}(T_1, S(T_1); T, K)
=
e^{-\rho(T-T_1)}
\mathbb{E}  \left[(S(T)-K)^+ \mid S(T_1)\right].
\end{equation}
Here, the inner option represents the value at time $T_1$ of the remaining investment opportunity with terminal payoff $(S(T)-K)^{+}$ at $T$, and $\rho$ denotes the adopted discount rate.

The time-0 value of the compound investment opportunity is therefore
\begin{equation}\label{3}
V(0,S(0))
=
e^{-\rho T_1}
\mathbb{E} 
\left[
\left( V_{\mathrm{inner}}(T_1, S(T_1); T, K) -K_1 \right)^+
\right].
\end{equation}

Except in the Black--Scholes setting, closed-form solutions are generally not available. In real options applications, geometric Brownian motion (GBM) has traditionally been adopted for tractability \cite{DixitPindyck1994,Trigeorgis1996}. 
However, empirical studies show that jump risk, heavy tails, and stochastic volatility can materially affect investment thresholds and option values \cite{AlvarezRakkolainen2010,Nishihara2016,Chronopoulos2018,LempaSaarinen2020}. 
This is relevant in staged R\&D and infrastructure projects, where uncertainty characteristics may differ between phases. In such environments, valuation accuracy depends not only on modeling flexibility, but also on computational methods capable of handling nested optionality without excessive numerical burden. Consequently, numerical schemes must balance modeling richness with computational tractability across multiple decision layers.

Many economically relevant models, including exponential Lévy and
jump-diffusion models, admit closed-form characteristic functions \cite{ContTankov2004}. This has motivated Fourier-based pricing techniques such as the transform approach of Carr and Madan \cite{CarrMadan1999}. Among these, the Fourier cosine (COS) method of Fang and Oosterlee \cite{FangOosterlee2008,FangOosterlee2009} provides an efficient framework for option valuation under such dynamics. Under suitable regularity conditions, and in particular when the relevant tails decay exponentially, the COS method exhibits exponential convergence; see \cite{Junike2024}.
The COS method is attractive in the compound setting because it avoids state-space discretization and operates directly in transform space. This makes it a natural candidate for handling nested expectations arising in compound option structures.
While the COS method has been applied to compound options, see, for example~\cite{Schyns2025}, existing implementations typically rely on numerical quadrature for the outer cosine coefficients, leading to additional computational cost.  More importantly, this quadrature introduces an additional approximation layer that can obscure spectral convergence and limit scalability in multi-stage settings. To the best of our knowledge, a fully analytic treatment of the outer COS coefficients for compound options has not yet been developed in the literature.

The main contribution of this paper is to derive a fully analytic compound COS formulation in which the outer cosine coefficients are evaluated in closed form by exploiting the trigonometric structure of the inner COS expansion. The resulting framework maintains the spectral accuracy of the COS method while improving computational efficiency at comparable accuracy. In contrast to quadrature-based implementations, the proposed method eliminates the numerical quadrature layer and therefore introduces no additional quadrature error into the compound recursion. Beyond its computational advantages, the analytic structure also provides a transparent mathematical formulation of the compound recursion.

An additional advantage of the proposed approach is its flexibility. Different independent-increment dynamics may be employed between successive decision dates without altering the analytic recursion structure. In particular, stage-specific diffusion and jump distributions can be incorporated by changing the increment characteristic function associated with the corresponding time interval. In simulation-based approaches such as Monte Carlo, introducing such stage-dependent dynamics within nested compound structures typically requires regression-based continuation estimation or nested simulation, while the simulation of jump processes adds further complexity. The transform-based COS recursion avoids these difficulties and allows heterogeneous risk dynamics across stages to be incorporated without modifying the underlying numerical scheme.

The remainder of this paper is organized as follows. Section~\ref{s2} develops the analytic compound COS framework. We formulate the compound pricing problem within the COS method, derive the closed-form evaluation of the outer cosine coefficients, extend the approach to multiple compound decision dates and provide a rigorous error analysis. Section~\ref{s3} demonstrates that the analytic COS structure extends beyond compound options by applying it to other contracts with single-boundary continuation payoffs, including chooser and Bermudan-type options. Section~\ref{s4} presents numerical experiments validating the method under the Black–Scholes model, studying convergence and computational efficiency, and illustrating applications under jump-diffusion dynamics and in real options valuation.

\section{Mathematical framework}\label{s2}

Let us briefly review the COS method introduced by \cite{FangOosterlee2008}, which forms the basis of the numerical method developed in this paper.

Let $ (\Omega, \mathcal{F}, (\mathcal{F}_t)_{t \geq 0}, \mathbb{P} )$ be a filtered probability space satisfying the usual conditions. We consider a positive stochastic process $ S\equiv (S(t))_{t \geq 0}$ representing the value of an underlying project or asset, and define the log-value process by
\begin{equation}
X(t)=\log S(t)
\end{equation}
and throughout the paper, we assume that $X\equiv (X(t))_{t \geq 0}$ is a Markov process.

In real-option applications, the dynamics of $S$ are typically specified under the physical probability measure $\mathbb{P}$, while future cash flows are discounted using a project discount rate $\rho$. The COS methodology requires only that the conditional characteristic function of the transition distribution of $X$ is available in closed form.

More precisely, for $0\le t<T$ and $X(t)=x$, define conditional characteristic function
\begin{equation}
\varphi_X(u;t,T,x)
:=
\mathbb E
\left[
e^{iuX(T)}
\,\middle|\,
X(t)=x
\right],\quad \forall u\in\mathbb R
\end{equation}
By the Markov property, conditioning on the information available at time $t$ can be reduced to conditioning on the current state $X(t)=x$.

Consider a contingent claim with payoff $g(X(T))$. Its value at time $t$ is given by
\begin{equation}\label{6a}
V(t,x)
=
e^{-\rho(T-t)}
\mathbb E
\left[
g(X(T))
\,\middle|\,
X(t)=x
\right].
\end{equation}

Suppose that the transition distribution of $X(T)$ given $X(t)=x$ admits a density, denoted by $f(y;T\mid x,t)$. Then \eqref{6a} can be written as
\begin{equation}\label{6}
V(t,x)
=
e^{-\rho(T-t)}
\int_{\mathbb R}
g(y)
f(y;T\mid x,t)
\,\mathrm dy.
\end{equation} 

Following the COS method, we truncate the integration domain to a finite interval $[a', b']$ containing most of the probability mass of the transition density. On this interval, the density is approximated by the first $N$ terms of its Fourier-cosine expansion,
\begin{equation}\label{7}
f(y;T\mid x,t)
\approx
\sum_{k=0}^{N-1}{}'
F_k(t,x)
\cos   \left( \frac{k\pi(y - a')}{b' - a'} \right),
\qquad
y\in[a', b'],
\end{equation}
where the prime indicates that the first term is multiplied by one-half, and $F_k(t,x)$ are approximations to the Fourier-cosine coefficients of $f$, computed directly from the characteristic function, 
\begin{equation}
F_k(t,x)
=
\frac{2}{b' - a'}
\Re
\left\{
\varphi_X
\left(
\frac{k\pi}{b' - a'};
t,T,x
\right)
e^{-ik\pi a'/(b' - a' )}
\right\} .
\end{equation}

After truncating the integral in \eqref{6} to $[a',b']$ and substituting the approximation \eqref{7}, we obtain the COS valuation formula, 
\begin{equation}
V(t,x)
\approx
e^{-\rho(T-t)}
\sum_{k=0}^{N-1}{}'
\Re
\left\{
\varphi_X
\left(
\frac{k\pi}{b'-a'};
t,T,x
\right)
e^{-ik\pi a'/(b'-a')}
\right\}   H_k
\end{equation}
where $H_k$ are the Fourier-cosine coefficients of the payoff function $g$ and given by
\begin{equation}
H_k
\coloneqq
\frac{2}{b' - a'}
\int_{a'}^{b'}
g(y)
\cos 
\left(
\frac{k\pi(y-a')}{b' - a'}
\right)
\,\mathrm dy.
\end{equation}

\subsection{COS method for compound option pricing}\label{sec:simplecompound}
\noeqref{2}

We now apply the COS methodology to the compound option pricing problem introduced in \eqref{1}--\eqref{3}. The valuation consists of two nested conditional expectations: the inner expectation \eqref{2} determines the value of the underlying option at the intermediate decision time $T_1$, while the outer expectation \eqref{3} prices the compound option at the initial time, with payoff of the form \eqref{1} at $T_1$. The objective is to preserve the transform-based structure of the COS method at both levels, thereby yielding an analytic pricing formulation that is free of quadrature.

Let $X(t) = \log(S(t))$ denote the log-project-value process. Then the inner option value corresponding to \eqref{2} can be written as 
\begin{equation}\label{C_inner}
V_{\mathrm{inner}}(x)
=
e^{-\rho(T-T_1)}
\int_{\mathbb R}
g(y)\,
f(y;T\mid x,T_1)
\,\mathrm dy,
\end{equation}
where $g(y)=(e^y-K)^+$ is the payoff at terminal time $T$.

Applying the COS expansion described previously, the transition density is approximated on a fixed truncated interval $[a',b']$. Define
\begin{equation}
\omega_k=\frac{k\pi}{b'-a'},
\qquad
k=0,\ldots,N_{\mathrm{in}}-1,
\end{equation}
and let
\begin{equation}
V_k
=
\frac{2}{b'-a'}
\int_{a'}^{b'}
(e^y-K)^+
\cos  \left(\omega_k(y-a')\right)
\,\mathrm dy,
\label{13}
\end{equation}
denote the cosine coefficients of the payoff function $g$. Note that the analytical expression for $V_k$ is available for the chosen $g(y)$.

The COS approximation of \eqref{C_inner} is therefore
\begin{equation}\label{C_inner_approx}
V_{\mathrm{inner}}^{\mathrm{COS}}(x)
=
e^{-\rho(T-T_1)}
\sum_{k=0}^{N_{\mathrm{in}}-1}{}'
\Re\!\left\{
\varphi_X(\omega_k;T_1,T,x)
e^{-i\omega_k a'}
\right\}
V_k.
\end{equation}

In the remainder of this paper, we specialize to models for which the log-process $X$ has independent increments. This class includes the Black--Scholes model and exponential Lévy models. Since the increment $X(T)-X\left(T_1\right)$ is independent of $X\left(T_1\right)$, the conditional characteristic function factorizes as
$$
\varphi_X\left(u ; T_1, T, x\right)=e^{i u x} \psi\left(u ; T_1, T\right)
$$
where
$$
\psi(u;T_1,T)
:=
\mathbb E
\left[
e^{iu(X(T)-X(T_1))}
\right]
$$
is the characteristic function of the increment $X(T)-X(T_1)$. Consequently,
\begin{equation}\label{C_inner_cos}
V_{\mathrm{inner}}^{\mathrm{COS}}(x)
= 
e^{-\rho(T-T_1)}
\sum_{k=0}^{N_{\mathrm{in}}-1}{}'
\Re\!\left\{
\psi(\omega_k;T_1,T)
e^{i\omega_k(x-a')}
\right\}
V_k.
\end{equation}

Using Euler's formula and expanding the real part yields the trigonometric representation
\begin{equation}\label{C_inner_cos_tri}
V_{\mathrm{inner}}^{\mathrm{COS}}(x)
=
\sum_{k=0}^{N_{\mathrm{in}}-1}{}'
\left(
A_k
\cos(\omega_k(x-a'))
-
B_k
\sin(\omega_k(x-a'))
\right),
\end{equation}
where 
\begin{equation}\label{C_inner_cos_AB}
A_k
\coloneqq
e^{-\rho(T-T_1)}
\Re  \left\{ \psi(\omega_k;T_1,T) \right\}
V_k,
\qquad
B_k
\coloneqq
e^{-\rho(T-T_1)}
\Im \left\{ \psi(\omega_k;T_1,T) \right\}
V_k.
\end{equation}
The coefficients $A_k$ and $B_k$ are well defined irrespective of whether the payoff coefficients $V_k$ are evaluated analytically or numerically. However, a closed-form expression for $V_k$ is required to obtain a fully quadrature-free implementation.

With the analytical expression for $V_{\mathrm{inner}}^{\mathrm{COS}}(x)$, we can obtain the compound option value, corresponding to equation \eqref{3}, by applying a second COS approximation to the outer option with a payoff
$
\left(V_{\mathrm{inner}}(x)-K_1\right)^{+}
$
at time $T_1$. Let $x_0=\log S(0)$ and define
$$
\nu_n=\frac{n \pi}{b-a}, \quad n=0, \ldots, N_{\mathrm{out}}-1 .
$$

The outer option value is then first approximated by
\begin{equation}\label{bar_C_outer}
\bar{V}_{\mathrm{outer}}^{\mathrm{COS}}(x_0)
= e^{-\rho T_1}   \sum_{n=0}^{N_{\mathrm{out}}-1}{}' 
\Re  \left\{ \psi(\nu_n ; 0, T_1) e^{i \nu_n(x_0-a)} \right\} \Bar{H}_n^{\mathrm{out}} 
\end{equation}
where
\begin{equation}\label{H_out_exact}
\Bar{H}_n^{\mathrm{out}} 
=\frac{2}{b-a}  \int_a^b \left( V_{\mathrm{inner}}(x) - K_1\right)^{+} \cos \left(\nu_n(x-a)\right) \mathrm{d} x . 
\end{equation}

Under the independent-increment assumption and monotonicity of the payoff $g$,  $V_{\mathrm{inner}}(x)$ is non-decreasing in $x$. We assume that $V_{\mathrm{inner}}(x)$ is strictly increasing on the outer truncation interval $[a, b]$ and that $V_{\mathrm{inner}}(a)< K_1 < V_{\mathrm{inner}}(b)$. It then follows that there exists a unique exercise boundary $x^* \in(a, b)$ satisfying
$$
V_{\mathrm{inner}} \left( x^* \right)=K_1 .
$$

In the COS implementation, the exact continuation value $V_{\mathrm{inner}}$ is replaced by its finite approximation $V_{\mathrm{inner}}^{\mathrm{COS}}$. We therefore define the numerical exercise boundary $\widetilde{x}^*$ as a solution of
\begin{equation}\label{cos_boundary}
V_{\mathrm{inner}}^{\mathrm{COS}} (\widetilde x^*) = K_1,
\end{equation} 
The computation of $\widetilde{x}^*$ is discussed in Section~\ref{sec:boundary}, while its existence is established in Section~\ref{sec:error}.
 
Throughout this paper, we assume that, whenever the boundary equation \eqref{cos_boundary} admits a solution $\widetilde{x}^* \in(a,b)$, the solution is unique and $V_{\mathrm{inner}}^{\mathrm{COS}}(x)-K_1$ changes sign only at $\widetilde{x}^*$. Under this single-boundary condition, the payoff at $T_1$ can be written as 
\begin{equation}\label{20}
\left(V_{\mathrm{inner}}^{\mathrm{COS}}(x)-K_1\right)^+
=
\begin{cases}
0, & x \leq \widetilde x^*,\\
V_{\mathrm{inner}}^{\mathrm{COS}}(x)-K_1, & x>\widetilde x^*.
\end{cases}
\end{equation}

Consequently, the outer payoff coefficient $\bar{H}_{n}^{\mathrm{out}}$ can be approximated by
\begin{equation}\label{H_out_reduced}
H_{n}^{\mathrm{out}}
=
\frac{2}{b-a}
\int_{\widetilde x^*}^{b}
\left(
V_{\mathrm{inner}}^{\mathrm{COS}}(x)-K_1
\right)
\cos\!\left(
\frac{n\pi(x-a)}{b-a}
\right)
\,\mathrm dx ,
\end{equation}
and we can replace $\bar{H}_{n}^{\mathrm{out}}$ in \eqref{bar_C_outer} with this $H_{n}^{\mathrm{out}}$ and define the final COS approximation $V_{\mathrm{outer}}^{\mathrm{COS}}(x_0)$,
\begin{equation}\label{C_outer}
V_{\mathrm{outer}}^{\mathrm{COS}}(x_0)
= e^{-\rho T_1}   \sum_{n=0}^{N_{\mathrm{out}}-1}{}' 
\Re  \left\{ \psi(\nu_n ; 0, T_1) e^{i \nu_n(x_0-a)} \right\} H_n^{\mathrm{out}} 
\end{equation}

The difference between $\bar{H}_{n}^{\mathrm{out}}$ and $H_{n}^{\mathrm{out}}$ arises from the COS approximation of the inner continuation value and the resulting displacement of the exercise boundary. These errors are accounted for explicitly in the error analysis in Section \ref{sec:error}.

Now, let us substitute the finite trigonometric representation \eqref{C_inner_cos_tri} into the coefficient \eqref{H_out_reduced} and therefore obtain
\begin{equation}\label{expanded_integral}
\begin{aligned}
H_{n}^{\mathrm{out}}
&=
\frac{2}{b-a}
\sum_{k=0}^{N_{\mathrm{in}}-1}{}'
\Bigg[
A_k
\int_{\widetilde x^*}^{b}
\cos\bigl(\omega_k(x-a')\bigr)
\cos\bigl(\nu_n(x-a)\bigr)
\,\mathrm dx
\\
&\hspace{4.5cm}
-
B_k
\int_{\widetilde x^*}^{b}
\sin\bigl(\omega_k(x-a')\bigr)
\cos\bigl(\nu_n(x-a)\bigr)
\,\mathrm dx
\Bigg]
\\
&\quad
-
\frac{2K_1}{b-a}
\int_{\widetilde x^*}^{b}
\cos\bigl(\nu_n(x-a)\bigr)
\,\mathrm dx
\\
&=
\sum_{k=0}^{N_{\mathrm{in}}-1}{}'
\left(
A_k I_{k,n}^{(c)}
-
B_k I_{k,n}^{(s)}
\right)
-
K_1 I_n^{(o)}   ,
\end{aligned}
\end{equation}
where
\begin{align*}
I_{k,n}^{(c)}
& \coloneqq
\frac{2}{b-a}
\int_{\widetilde x^*}^{b}
\cos (\omega_k(x-a'))
\cos (\nu_n(x-a))
\,\mathrm dx,
\\
I_{k,n}^{(s)}
& \coloneqq
\frac{2}{b-a}
\int_{\widetilde x^*}^{b}
\sin (\omega_k(x-a'))
\cos (\nu_n(x-a))
\,\mathrm dx,
\\
I_n^{(o)}
& \coloneqq
\frac{2}{b-a}
\int_{\widetilde x^*}^{b}
\cos (\nu_n(x-a))
\,\mathrm dx.
\end{align*}
The explicit formulas for these integral terms are derived in Appendix~\ref{app:detail_Hout}.

Indeed, one of the main contributions of this paper is to exploit the trigonometric representation \eqref{C_inner_cos_tri} to derive closed-form expressions for the outer payoff coefficients $H_n^{\mathrm{out}}$. Once the numerical boundary $\widetilde{x}^*$ has been determined, the outer COS coefficients can therefore be evaluated without numerical quadrature, eliminating both the additional quadrature layer and the associated quadrature error.

Our proposed approach requires solving only one scalar nonlinear equation for $\widetilde{x}^*$ and then evaluating closed-form trigonometric expressions. Compared with the quadrature-based evaluation of the outer coefficients, whose computational cost is of order
$
\mathcal{O}\left(n_q\left(N_{\mathrm{in}}+N_{\mathrm{out}}\right)\right) ,
$ 
where $n_{\mathrm{q}}$ denotes the number of quadrature points, the closed-form evaluation has complexity 
$\mathcal{O}\left(N_{\mathrm{out}} N_{\mathrm{in}}\right)$, 
and the additional cost of solving the scalar boundary equation is typically negligible relative to the outer COS summations.

\subsection{Extension to multiple compound decision dates}
\label{sec:multi_compound}

We now extend the analytic COS methodology to the case of multiple compound decision dates. For illustration, we consider multiple compound call options with the same payoff structure as in the simple two-stage case. The methodology, however, remains applicable to a broader class of payoff functions, as discussed in the subsequent subsections.

Let
\begin{equation}
0 \coloneqq T_0 < T_1 < T_2 < \dots < T_m < T_{m+1} \coloneqq T,
\end{equation}
where $T$ denotes the final maturity. At each intermediate decision date $T_i$, $1\leq i\leq m$, the holder has the right, but not the obligation, to acquire the continuation value by paying the strike $K_i$. At the final maturity $T$, the payoff is a European call option with strike $K_{m+1}$.

We use the following convention. For the COS evaluation of the conditional expectation from $T_i$ to $T_{i+1}$, the integration variable is $X_{i+1} \equiv X(T_{i+1})$. Hence, the corresponding truncation interval is denoted by $[a_{i+1},b_{i+1}]$, and the number of cosine terms is denoted by $N_{i+1}$. The quantities $N_{i+1}$ and $[a_{i+1},b_{i+1}]$ may vary with $i$, allowing the spectral resolution and truncation range to adapt to each time step.

For $i=1,\ldots,m+1$, define
\begin{equation}
\omega_k^{(i)}
=
\frac{k\pi}{b_i-a_i},
\qquad
k= 0, 1, \ldots, N_i - 1. 
\end{equation}

Let $V^{(i)}(x_i)$ denote the continuation value evaluated at time $T_i$ when $X(T_i) = x_i$. For $i=m,m-1,\ldots,0$, the COS approximation for $V^{(i)}(x_i)$ takes the form
\begin{equation}\label{multi_COS_recursion}
V^{(i),\mathrm{COS}}(x_i)
=
e^{-\rho(T_{i+1}-T_i)}
\sum_{k=0}^{N_{i+1}-1}{}'
\Re
\left\{
\psi
\left(
\omega_k^{(i+1)};T_i,T_{i+1}
\right)
e^{i\omega_k^{(i+1)}(x_i-a_{i+1})}
\right\}
H_k^{(i+1)}.
\end{equation}

Here, for the case $i=m$, the coefficients $H_k^{(m+1)}$ are the usual cosine coefficients of the terminal call payoff,
\begin{equation}\label{terminal_payoff_coeff}
H_k^{(m+1)}
=
\frac{2}{b_{m+1}-a_{m+1}}
\int_{a_{m+1}}^{b_{m+1}}
\bigl(e^{x_{m+1}}-K_{m+1}\bigr)^+
\cos
\left(
\omega_k^{(m+1)}(x_{m+1}-a_{m+1})
\right)
\,\mathrm d x_{m+1}.
\end{equation}

Then, for $i=m-1,\ldots,0$, since the payoff at time $T_{i+1}$ is of the form $\bigl(V^{(i+1)}(x_{i+1})-K_{i+1}\bigr)^+$, we can substitute $V^{(i+1),\mathrm{COS}}(x_{i+1})$ that we have obtained, into the cosine coefficients $H_k^{(i+1)}$,
\begin{equation}\label{multi_H_coeff}
H_k^{(i+1)}
=
\frac{2}{b_{i+1}-a_{i+1}}
\int_{\widetilde{x}_{i+1}^*}^{b_{i+1}}
\bigl(V^{(i+1),\mathrm{COS}}(x_{i+1})-K_{i+1}\bigr)
\cos\!\left(
\omega_k^{(i+1)}(x_{i+1}-a_{i+1})
\right)
\,\mathrm d x_{i+1}.
\end{equation}
where we have used the numerical boundary $\widetilde x_{i+1}^*$ to simplify the integral, and it is obtained by solving $V^{(i+1),\mathrm{COS}}(\widetilde x_{i+1}^*)=K_{i+1}$.

Using Euler's formula, equation \eqref{multi_COS_recursion} for $V^{(i),\mathrm{COS}}(x_i)$ can be rewritten as the following trigonometric expansion
\begin{equation}\label{m2}
V^{(i),\mathrm{COS}}(x_i)
=
\sum_{k=0}^{N_{i+1}-1}{}'
\left(
A_k^{(i)}
\cos \bigl( \omega_k^{(i+1)}(x_i-a_{i+1}) \bigr)
-
B_k^{(i)}
\sin \bigl( \omega_k^{(i+1)}(x_i-a_{i+1}) \bigr)
\right),
\end{equation}
where we define
\begin{equation}\label{multi_AB_coeff}
A_k^{(i)}
\coloneqq
e^{-\rho(T_{i+1}-T_i)}
\Re
\left\{
\psi(\omega_k^{(i+1)};T_i,T_{i+1})
\right\}
H_k^{(i+1)}
,\quad
B_k^{(i)}
\coloneqq
e^{-\rho(T_{i+1}-T_i)}
\Im
\left\{
\psi(\omega_k^{(i+1)};T_i,T_{i+1})
\right\}
H_k^{(i+1)}.
\end{equation}

This trigonometric representation is essential for the analytic recursion.
Indeed, to compute $H_k^{(i+1)}$ for $i=m-1, \ldots, 0$, one substitutes the
trigonometric expansion of $V^{(i+1),\mathrm{COS}}$ into \eqref{multi_H_coeff}. The resulting integrals have the same structure as in the single-compound case and can be evaluated in closed form using trigonometric product identities. The corresponding formulas for the multi-stage coefficients are provided in Appendix~\ref{app:multistage_analytic_coeffs}.

The computational procedure can be summarized as follows. First, at $i=m$,
one computes $V^{(m),\mathrm{COS}}$ using the closed-form terminal payoff
coefficients $H_k^{(m+1)}$. These coefficients are available explicitly for a
European call payoff as well as several other types of payoffs. Then, for $i=m-1$ case, we first solve $V^{(i+1),\mathrm{COS}}(\widetilde x_{i+1}^*)=K_{i+1}$ for the numerical exercise boundary. This boundary determines the integration range in \eqref{multi_H_coeff}, from which the coefficients $H_k^{(i+1)}$ are obtained analytically. Substituting these coefficients into \eqref{multi_COS_recursion} then gives $V^{(i),\mathrm{COS}}$. Repeating this procedure backwards to $i=0$ yields the price $V^{\mathrm{COS}}(x_0) \coloneqq V^{(0),\mathrm{COS}}(x_0)$ at initial time.

As each compound step produces a trigonometric representation of the
continuation value, the analytic COS structure is preserved under backwards
recursion:
\begin{equation}
V^{(m),\mathrm{COS}}(x_m)
\rightarrow
V^{(m-1),\mathrm{COS}}(x_{m-1})
\rightarrow
\cdots
\rightarrow
V^{(0),\mathrm{COS}}(x_0).
\end{equation}

In particular, each stage requires only the solution of a scalar boundary equation and the analytic evaluation of the corresponding cosine coefficients, without introducing an additional quadrature or state-space discretization. For $m$ compound decision dates, the computational complexity of the analytic coefficient construction is
$\mathcal \mathcal{O}\left(\sum_{i=1}^{m} N_iN_{i+1}\right)$.
By comparison, an efficient quadrature-based implementation requires
$
\mathcal  \mathcal{O} \left( n_{\mathrm q}\sum_{i=1}^{m}(N_i+N_{i+1})  \right),
$
where $n_{\mathrm q}$ denotes the number of quadrature points. Although the difference in asymptotic complexity depends on the relative sizes of $N_i$ and $n_{\mathrm q}$, the analytic formulation avoids the additional quadrature error at every compounding stage. More importantly, it preserves the spectral approximation structure of the COS method, leading to the corresponding spectral convergence behavior under suitable regularity conditions, as discussed in the later Section~\ref{sec:error}.

Thus, the analytic compound COS formulation provides an efficient and fully quadrature-free framework for pricing multi-layer compound options under models whose log-process has tractable increment characteristic functions. The formulation also accommodates stage-dependent increment distributions, since a different increment characteristic function may be used on each interval $[T_i,T_{i+1}]$, provided that the independent-increment structure is preserved.

\subsection{Efficient computation of the exercise boundary} \label{sec:boundary}

We now discuss the computation of the numerical exercise boundary $\widetilde x^*$ defined by \eqref{cos_boundary}, together with its extension to multiple compound exercise dates.

Recall the simple compound case introduced in Section \ref{sec:simplecompound}.
By the trigonometric representation \eqref{C_inner_cos_tri}, the value $V_{\mathrm{inner}}^{\operatorname{COS}}(x)$ can be evaluated in $\mathcal{O}\left(N_{\mathrm{in}}\right)$ operations without numerical integration. Differentiating the finite expansion term by term gives
\begin{equation}\label{inner_derivative}
\frac{\mathrm d}{\mathrm dx}V_{\mathrm{inner}}^{\mathrm{COS}}(x)
=
\sum_{k=0}^{N_{\mathrm{in}}-1}{}'
\left(
-A_k\omega_k\sin\bigl(\omega_k(x-a')\bigr)
-
B_k\omega_k\cos\bigl(\omega_k(x-a')\bigr)
\right).
\end{equation}
The prime does not affect the $k=0$ contribution in the derivative, since
$\omega_0=0$, but we keep it for consistency with \eqref{C_inner_cos_tri}.
Thus, both the COS continuation function and its derivative can be evaluated
at essentially the same computational cost.

For a bracketing method, a strict sign change of $V_{\mathrm{inner}}^{\mathrm{COS}}(x)-K_1$ is needed. We therefore choose the outer truncation interval $[a,b]$ such that
\begin{equation}
V_{\mathrm{inner}}^{\mathrm{COS}}(a)<K_1<
V_{\mathrm{inner}}^{\mathrm{COS}}(b).
\label{boundary_bracket}
\end{equation}
Under the single-boundary assumption introduced in Section~\ref{sec:simplecompound}, this provides a valid bracket for the numerical exercise boundary $\widetilde{x}^*$. If \eqref{boundary_bracket} is not satisfied, the interval may be enlarged until a valid bracket is obtained.

A convenient initial guess $x^{(0)}$ can be obtained by approximating the inner continuation value by a discounted intrinsic-value proxy. Equating this proxy to the outer strike gives 
\begin{equation}
e^{-\rho(T-T_1)}   \left(e^{x^{(0)}}-K\right)^{+}=K_1,
\end{equation}
and hence
\begin{equation}\label{initial_guess}
x^{(0)} = \log  \left( K+K_1e^{\rho(T-T_1)} \right).
\end{equation}
This estimate is used only as an initial value for the nonlinear solver.

Using the derivative \eqref{inner_derivative}, Newton's method for solving the COS boundary equation is given by  
\begin{equation}\label{newton_update}
x^{(\ell+1)}
=
x^{(\ell)}
-
\frac{
V_{\mathrm{inner}}^{\mathrm{COS}}(x^{(\ell)})-K_1
}{
\bigl( V_{\mathrm{inner}}^{\mathrm{COS}} \bigr)' (x^{(\ell)})
}.
\end{equation}
Since $V_{\mathrm{inner}}^{\mathrm{COS}}$ is a finite trigonometric polynomial, it is smooth. If
$
\bigl( V_{\mathrm{inner}}^{\mathrm{COS}} \bigr)'  (\widetilde{x}^*) \neq 0,
$
Newton's method converges quadratically when the initial iterate is sufficiently close to $\widetilde{x}^*$. Moreover, the iteration is terminated when the absolute residual falls below a prescribed tolerance.

To ensure global robustness, we use a safeguarded Newton method. Newton updates are restricted to the current bracketing interval initialized by \eqref{boundary_bracket}. If a proposed update lies outside the current bracket, a bisection step is performed instead. Alternatively, Brent's method may be used; it combines bracketing with secant and inverse quadratic interpolation steps and does not require the analytic derivative. Since each function evaluation costs only $\mathcal O(N_{\mathrm{in}})$, the boundary computation is typically negligible relative to the outer COS summations, whose cost is $\mathcal O(N_{\mathrm{out}}N_{\mathrm{in}})$.

Finally, we consider the case of multiple compound exercise dates. In this case, the exercise boundary computed at the adjacent later stage $i+1$ may provide a useful warm start for the current stage. In practice, one may take
$
x_i^{(0)} = \widetilde x_{i+1}^*
$
as an initial guess for computing $\widetilde x_{i}^*$.
The same safeguarded Newton procedure or Brent’s method can then be applied at each stage.

\subsection{Choice of truncation intervals across compound stages}
\label{sec:interval}

In the COS method, integrals over $\mathbb{R}$ are approximated by integrals over a finite truncation interval. For single-maturity problems, a widely used choice is the cumulant-based interval proposed in \cite{FangOosterlee2008}, in which the truncation interval $[a,b]$ is constructed from the first, second, and fourth cumulants of the log-value distribution. In the recursive compound COS formulation of Section~\ref{sec:multi_compound}, stage-dependent truncation intervals $[a_i,b_i]$ are required for the successive state variables $X_i:=X(T_i)$, for $i=1,\ldots,m+1$.

Under the independent-increment assumption, a natural recursive construction is based on the cumulants of the increments
\begin{equation}
\Delta X_i:=X(T_{i+1})-X(T_i),
\qquad
i=0,1,\ldots,m.
\end{equation}
Let $c_1^{(i)}$, $c_2^{(i)}$, and $c_4^{(i)}$ denote the first, second, and fourth cumulants of $\Delta X_i$, respectively. We assume that these cumulants are finite. Following the classical COS construction, define the half-width
\begin{equation}
d_i
\coloneqq
L\sqrt{c_2^{(i)}+\sqrt{c_4^{(i)}}},  \quad \forall\ L\in[8, 12].
\end{equation}

Let $[a_0, b_0]$ be a chosen initial interval for $X(T_0)$. Using the increment representation
\begin{equation}
X(T_{i+1}) = X(T_i)+\Delta X_i, \quad i=0,\ldots, m 
\end{equation}
and given a truncation interval $[a_i,b_i]$ for $X(T_i)$, we approximate the relevant range of the increment $\Delta X_i$ by $[ c_1^{(i)}-d_i,  c_1^{(i)}+d_i ]$. The truncation interval $[a_{i+1},b_{i+1}]$ for $X(T_{i+1})$ is then obtained by adding these two intervals:
\begin{equation} \label{eq:recursive_interval}
[a_{i+1}, b_{i+1}]
=
\left[
a_i+c_1^{(i)}-d_i,\,
b_i+c_1^{(i)}+d_i
\right].
\end{equation}
This recursion can be applied successively to construct the intervals $[a_{i}, b_{i}]$, from $i=1$ to $i=m+1$.  Moreover, as the cumulants depend on both the time increment $T_{i+1}-T_i$ and the parameters of the underlying model, the constructed intervals adapt naturally to different distributions of the state dynamics.

As an alternative to the forward cumulant propagation, the numerical exercise boundaries obtained for the recursive compound formulation in Section~\ref{sec:multi_compound}, can be used to construct boundary-anchored truncation intervals backward in time.

Let $[a_0, b_0]$ be an initial interval again, and we directly construct the intervals $[a_{m+1}, b_{m+1}]$ for $X(T_{m+1}) = X(T)$ using cumulants of the increment $X(T_{m+1}) - X(0) $; thus, we can obtain the COS approximation $V^{(m), \mathrm{COS}}$. At the decision date $T_{m}$, we have a payoff 
$
\bigl( V^{(m), \mathrm{COS}}(x_{m}) - K_{m} \bigr)^+
$
which vanishes when $x_{m}\le \widetilde{x}_{m}^{*}$, and the boundary can be obtained by solving 
$$
V^{(m), \mathrm{COS}}  \left( \widetilde{x}_{m}^{*} \right)
= K_{m}.
$$

Hence, only the region to the right of the exercise boundary contributes to the payoff COS coefficients. This observation motivates the boundary-anchored interval for $X(T_{m})$,
\begin{equation}\label{eq:boundary_interval}
[a_{m},  b_{m}]
\coloneqq
\bigl[
\widetilde{x}_{m}^{*} - \delta_{m}, 
\widetilde{x}_{m}^{*} + \Delta_{m}
\bigr],
\end{equation}
where
\begin{equation}
\delta_{i}
\coloneqq
\kappa\sqrt{c_2^{(i-1)}},
\qquad
\Delta_{i}
\coloneqq
c_1^{(i-1)}
+
L\sqrt{c_2^{(i-1)}+\sqrt{c_4^{(i-1)}}}, \quad \forall \, i = m, \ldots, 1,
\end{equation}
with $\kappa\in[1,2]$ and $L\in[8,12]$. With the interval $[a_m, b_m]$, we can obtain the COS approximation $V^{(m-1), \mathrm{COS}}$, and therefore we can repeat the above procedure to obtain all the intervals $[a_{i},b_{i}]$, from $m$ to $1$.

Because the continuation value is already represented by its finite trigonometric expansion, the numerical exercise boundary is naturally computed during the backward recursion. Consequently, the proposed boundary-anchored interval requires no additional boundary computation. By excluding part of the region in which the payoff vanishes identically, it reduces the effective integration domain. As a result, it may improve the spatial resolution of the payoff kink for a fixed number of cosine terms and, in practice, can reduce the number of terms required to achieve a prescribed accuracy.

For models with independent increments, the cumulants of each increment can be obtained directly from derivatives of the logarithm of the corresponding characteristic function evaluated at the origin. Consequently, the stage-dependent truncation intervals can be constructed efficiently once the characteristic function is available in closed form.

An alternative interval construction was proposed in \cite{junike2022cos}, where the truncation interval is determined from rigorous tail bounds derived using Markov's inequality together with moments obtained from derivatives of the characteristic function. This approach provides explicit control of the truncation error and may be advantageous when the classical cumulant-based interval is insufficiently wide.

\subsection{Error analysis of the compound COS formulation}
\label{sec:error}

We now investigate the approximation error of the recursive compound COS formulation. Since the continuation representations and the payoff cosine coefficients are evaluated analytically throughout the backward recursion, the proposed method introduces neither numerical quadrature nor interpolation errors. The remaining errors arise from the truncation of the integration intervals, the finite cosine expansions, and the recursive replacement of the exact continuation values by their COS approximations.

For $i=0,\ldots,m$, let $I_i \coloneqq [a_i,b_i]$, $\Delta T_i \coloneqq T_{i+1}-T_i$,
and define the exact conditional-expectation operator
\begin{equation}
\mathcal T_i h(x_i)
\coloneqq
e^{-\rho\Delta T_i}
\mathbb E
\left[
h(X_{i+1})
\,\middle|\,
X_i=x_i
\right],
\end{equation}
where $h$ is a given payoff function. We denote by $\mathcal T_{i,N_{i+1}}^{\mathrm{COS}}$ the corresponding COS approximation based on the truncation interval $I_{i+1}$ and $N_{i+1}$ cosine terms.

For convenience, we define the exact stagewise payoff functions as in Section~\ref{sec:multi_compound} by,
\begin{equation}
G^{(m+1)}(x)
\coloneqq
\bigl(e^x-K_{m+1}\bigr)^+
,\quad
G^{(i+1)}(x)
\coloneqq
\bigl(V^{(i+1)}(x)-K_{i+1}\bigr)^+ 
,\qquad  i=m-1,\ldots,0.
\end{equation}
Although the analysis below is presented for this compound-call payoff structure, the theoretical results are not restricted to this particular form. Their extension to more general stagewise payoff transformations is discussed in Remark~\ref{remark:general_payoff}.

Hence, the exact continuation values satisfy
\begin{equation}
V^{(i)}
=
\mathcal T_i G^{(i+1)},
\qquad
i=0,\ldots,m.
\end{equation}

To distinguish the one-stage COS approximation error from the recursively propagated error, we introduce the auxiliary stagewise approximation
\begin{equation}
\overline V^{(i),\mathrm{COS}}
\coloneqq
\mathcal T_{i,N_{i+1}}^{\mathrm{COS}}
G^{(i+1)}.
\label{eq:auxiliary_stagewise_cos}
\end{equation}
Thus, $\overline V^{(i),\mathrm{COS}}$ is obtained by applying a single COS approximation step to the exact payoff $G^{(i+1)}$. In particular, when $i<m$, its payoff coefficients are conceptually constructed from the exact continuation value $V^{(i+1)}$ and the exact exercise boundary $x_{i+1}^*$. This auxiliary quantity is introduced only for the error analysis and is not computed by the recursive numerical algorithm.

By contrast, the actual recursive COS approximation satisfies
\begin{equation}
V^{(i),\mathrm{COS}}
=
\mathcal T_{i,N_{i+1}}^{\mathrm{COS}}
G^{(i+1),\mathrm{COS}},
\label{eq:actual_recursive_cos}
\end{equation}
where
\begin{equation}
G^{(m+1),\mathrm{COS}}
=
G^{(m+1)}
,\qquad
G^{(i+1),\mathrm{COS}}(x)
\coloneqq
\bigl(
V^{(i+1),\mathrm{COS}}(x)-K_{i+1}
\bigr)^+  
,\quad i=m-1,\ldots,0.
\end{equation}
 
The distinction between $\overline V^{(i),\mathrm{COS}}$ and $V^{(i),\mathrm{COS}}$ separates the error introduced by a single COS approximation from the error propagated through the recursive replacement of the continuation values.

\paragraph{Stagewise COS approximation assumption.}

The convergence properties of a single COS approximation have been studied extensively in the classical COS literature; see, for example, \cite{FangOosterlee2008,Junike2024,wang2025note}. Rather than imposing specific regularity conditions on the transition densities, we formulate the one-stage COS error directly at the operator level.

For each $i=0,\ldots,m$, assume that there exists a function
$\delta_i:\mathbb N\to[0,\infty)$ satisfying
\begin{equation}
\delta_i(N)\longrightarrow 0,
\qquad
N\longrightarrow\infty,
\end{equation}
such that
\begin{equation}
\left\|
\mathcal T_iG^{(i+1)}
-
\mathcal T_{i,N}^{\mathrm{COS}}G^{(i+1)}
\right\|_{\infty,I_i}
\leq
\delta_i(N)
+
\tau_i,
\label{eq:stagewise_cos_assumption}
\end{equation}
where
\begin{equation}
\tau_i
\coloneqq
e^{-\rho\Delta T_i}
\sup_{x\in I_i}
\left|
\int_{\mathbb R\setminus I_{i+1}}
G^{(i+1)}(y)
f_i(y\mid x)
\,\mathrm dy
\right|
\label{eq:stagewise_truncation_error}
\end{equation}
denotes the payoff-weighted truncation error, and $f_i(y\mid x)$ is the conditional density of $X_{i+1}$ given $X_i=x$.

The quantity $\delta_i(N)$ represents the one-stage COS approximation error associated with the finite cosine representation on the truncated domain. Under appropriate regularity and characteristic function decay conditions, classical COS convergence results yield spectral rates of the form
\begin{equation}
\delta_i(N)
\leq
M_i e^{-\eta_iN},
\qquad
M_i,\eta_i>0.
\label{eq:stagewise_spectral_rate}
\end{equation}
We will use this stronger condition only when deriving the spectral convergence rate of the full recursive scheme.

We next establish the stability of the stagewise COS operator.

\begin{lemma}[Uniform stability of the stagewise COS operator]
\label{lem:cos_operator_stability}

Fix $i\in\{0,\ldots,m\}$ and assume that
\begin{equation}
\sum_{k=1}^{\infty}
\left|
\psi\!\left(
\omega_k^{(i+1)};T_i,T_{i+1}
\right)
\right|
<
\infty.
\label{eq:sampled_cf_summability}
\end{equation}
Then, for every $N\geq1$ and all bounded functions
$h,\widetilde h:I_{i+1}\to\mathbb R$,
\begin{equation}
\left\|
\mathcal T_{i,N}^{\mathrm{COS}}h
-
\mathcal T_{i,N}^{\mathrm{COS}}\widetilde h
\right\|_{\infty,I_i}
\leq
\Lambda_i
\left\|
h-\widetilde h
\right\|_{\infty,I_{i+1}},
\label{eq:cos_operator_stability}
\end{equation}
where
\begin{equation}
\Lambda_i
\coloneqq
e^{-\rho\Delta T_i}
\left[
1+
\frac{4}{\pi}
\sum_{k=1}^{\infty}
\left|
\psi\!\left(
\omega_k^{(i+1)};T_i,T_{i+1}
\right)
\right|
\right].
\label{eq:cos_stability_constant}
\end{equation}
In particular, $\Lambda_i$ is independent of the number $N$ of cosine terms.

\end{lemma}

\begin{proof}

For a bounded function $h:I_{i+1}\to\mathbb R$, define its cosine
coefficients by
\begin{equation}
H_k^{(i+1)}(h)
\coloneqq
\frac{2}{b_{i+1}-a_{i+1}}
\int_{a_{i+1}}^{b_{i+1}}
h(y)
\cos\!\left(
\omega_k^{(i+1)}(y-a_{i+1})
\right)
\,\mathrm dy.
\label{eq:general_stage_payoff_coeff}
\end{equation}
The stagewise COS operator can then be written as
\begin{equation}
\begin{aligned}
\mathcal T_{i,N}^{\mathrm{COS}}h(x)
&=
e^{-\rho\Delta T_i}
\sum_{k=0}^{N-1}{}'
\Re\!\left\{
\psi\!\left(
\omega_k^{(i+1)};T_i,T_{i+1}
\right)
e^{i\omega_k^{(i+1)}(x-a_{i+1})}
\right\}
H_k^{(i+1)}(h).
\end{aligned}
\label{eq:general_stage_cos_operator}
\end{equation}
Hence, the COS operator is linear. 

Let $g\coloneqq h-\widetilde h$. Then
$$
\mathcal T_{i,N}^{\mathrm{COS}}h
-
\mathcal T_{i,N}^{\mathrm{COS}}\widetilde h
=
\mathcal T_{i,N}^{\mathrm{COS}}g.
$$

For $k=0$, the cosine coefficient satisfies
\begin{equation}
\frac{1}{2}
\left|
H_0^{(i+1)}(g)
\right|
\leq
\|g\|_{\infty,I_{i+1}}.
\label{eq:zero_cosine_coefficient_bound}
\end{equation}

For $k\geq1$,
\begin{equation}
\begin{aligned}
\left|
H_k^{(i+1)}(g)
\right|
&\leq
\frac{2}{b_{i+1}-a_{i+1}}
\|g\|_{\infty,I_{i+1}}
\int_{a_{i+1}}^{b_{i+1}}
\left|
\cos\!\left(
\omega_k^{(i+1)}(y-a_{i+1})
\right)
\right|
\,\mathrm dy
=
\frac{4}{\pi}
\|g\|_{\infty,I_{i+1}}.
\end{aligned}
\label{eq:nonzero_cosine_coefficient_bound}
\end{equation}

Using the fact that 
$
\psi(0;T_i,T_{i+1})=1
$,
$
|\Re(z)|\leq |z|
$,
and 
$
|e^{i\theta}|=1
$,
we obtain for every $x\in I_i$,
\begin{equation}
\begin{aligned}
\left|
\mathcal T_{i,N}^{\mathrm{COS}}g(x)
\right|
&\leq
e^{-\rho\Delta T_i}
\Bigg[
\frac{1}{2}
\left|
H_0^{(i+1)}(g)
\right|
+
\sum_{k=1}^{N-1}
\left|
\psi\!\left(
\omega_k^{(i+1)};T_i,T_{i+1}
\right)
\right|
\left|
H_k^{(i+1)}(g)
\right|
\Bigg]
\\
&\leq
e^{-\rho\Delta T_i}
\Bigg[
1+
\frac{4}{\pi}
\sum_{k=1}^{N-1}
\left|
\psi\!\left(
\omega_k^{(i+1)};T_i,T_{i+1}
\right)
\right|
\Bigg]
\|g\|_{\infty,I_{i+1}}.
\end{aligned}
\end{equation}
By \eqref{eq:sampled_cf_summability},
\[
\sum_{k=1}^{N-1}
\left|
\psi\!\left(
\omega_k^{(i+1)};T_i,T_{i+1}
\right)
\right|
\leq
\sum_{k=1}^{\infty}
\left|
\psi\!\left(
\omega_k^{(i+1)};T_i,T_{i+1}
\right)
\right|.
\]
Taking the supremum over $x\in I_i$ gives
\eqref{eq:cos_operator_stability}.

\end{proof}

\begin{remark}
The summability condition \eqref{eq:sampled_cf_summability} is satisfied, for
example, when the increment characteristic function decays sufficiently fast.
A sufficient condition is
\[
\left| \psi(u;T_i,T_{i+1}) \right|
\leq  \beta_i(1+|u|)^{-p_i},  \qquad  p_i>1,
\]
for some constant $\beta_i>0$. Alternatively, it is sufficient to have
exponential-type decay of the form
\[
\left|
\psi(u;T_i,T_{i+1})
\right|
\leq
\beta_i e^{-\alpha_i|u|^{q_i}},
\qquad
\alpha_i,q_i,\beta_i>0.
\]
These conditions are satisfied, for example, by the log-increment characteristic functions of the GBM model and the classical Merton jump-diffusion model. Hence, both models satisfy \eqref{eq:sampled_cf_summability} for fixed truncation intervals. If the truncation intervals depend on the expansion sizes, we additionally assume that the resulting stability constants remain uniformly bounded with respect to the expansion sizes and the corresponding interval choices.
\end{remark}

\begin{theorem}[Recursive propagation and spectral convergence of COS errors]
\label{thm:recursive_cos_error}

Suppose that the stagewise approximation condition \eqref{eq:stagewise_cos_assumption} holds for each $i=0,\ldots,m$, and that the assumptions of Lemma~\ref{lem:cos_operator_stability} are satisfied with constants $\Lambda_i$ independent of the expansion sizes.

Define
\begin{equation}
\varepsilon_i
\coloneqq
\left\|
V^{(i)}
-
V^{(i),\mathrm{COS}}
\right\|_{\infty,I_i},
\qquad
i=0,\ldots,m.
\label{eq:recursive_error_definition}
\end{equation}
Then
\begin{equation}
\varepsilon_i
\leq
\sum_{j=i}^{m}
\left(
\prod_{\ell=i}^{j-1}\Lambda_\ell
\right)
\left[
\delta_j(N_{j+1})
+
\tau_j
\right],
\qquad
i=0,\ldots,m,
\label{eq:recursive_error_bound}
\end{equation}
where an empty product is understood to be equal to one.

In particular, suppose that the stagewise COS approximation errors satisfy
\begin{equation}
\delta_i(N)
\leq
M_i e^{-\eta_iN},
\qquad
i=0,\ldots,m,
\label{eq:stagewise_spectral_assumption}
\end{equation}
for constants $M_i,\eta_i>0$ independent of the expansion sizes. Define
$
N_{\min}
\coloneqq
\min_{1\leq j\leq m+1}N_j
$
and
$
\eta
\coloneqq
\min_{0\leq i\leq m}\eta_i
$,
then there exist constants $C>0$ and $C_{\mathrm{tr}}>0$, independent of
the expansion sizes, such that
\begin{equation}
\left|
V^{(0)}(x_0)
-
V^{(0),\mathrm{COS}}(x_0)
\right|
\leq
C e^{-\eta N_{\min}}
+
C_{\mathrm{tr}}
\max_{0\leq i\leq m}\tau_i,
\label{eq:global_recursive_cos_error}
\end{equation}

Consequently, if the truncation intervals are selected such that
\begin{equation}
\max_{0\leq i\leq m}\tau_i
=
\mathcal O\!\left(
e^{-\eta N_{\min}}
\right),
\label{eq:truncation_cos_balance}
\end{equation}
and the corresponding stability constants remain uniformly bounded with respect to the expansion sizes, then
\begin{equation}
\left|
V_0
-
V_0^{\mathrm{COS}}
\right|
=
\mathcal{O} \left(
e^{-\eta N_{\min}}
\right).
\end{equation}
Hence, under the stagewise spectral convergence condition, the recursive compound COS formulation preserves the spectral convergence rate of the individual COS approximation steps.

\end{theorem}

\begin{proof}

At the final continuation stage $i=m$, the terminal payoff is known exactly, so that $G^{(m+1),\mathrm{COS}}  =  G^{(m+1)}$. Consequently,
\[
V^{(m),\mathrm{COS}}  =  \overline V^{(m),\mathrm{COS}}.
\]

The stagewise approximation condition therefore gives
\begin{equation}
\varepsilon_m  \leq  \delta_m(N_{m+1})  +  \tau_m.
\end{equation}

Now consider $i<m$. Adding and subtracting the auxiliary approximation
$\overline V^{(i),\mathrm{COS}}$ gives
\begin{equation}
\begin{aligned}
\varepsilon_i
\leq
\left\|
V^{(i)}
-
\overline V^{(i),\mathrm{COS}}
\right\|_{\infty,I_i}
+
\left\|
\overline V^{(i),\mathrm{COS}}
-
V^{(i),\mathrm{COS}}
\right\|_{\infty,I_i}.
\end{aligned}
\label{eq:error_decomposition}
\end{equation}

By \eqref{eq:stagewise_cos_assumption}, the first term satisfies
\begin{equation}
\left\|
V^{(i)}
-
\overline V^{(i),\mathrm{COS}}
\right\|_{\infty,I_i}
\leq
\delta_i(N_{i+1})
+
\tau_i.
\label{eq:first_stagewise_error_term}
\end{equation}

For the second term, Lemma~\ref{lem:cos_operator_stability} gives
\begin{equation}
\left\|
\overline V^{(i),\mathrm{COS}}
-
V^{(i),\mathrm{COS}}
\right\|_{\infty,I_i}
\leq
\Lambda_i
\left\|
G^{(i+1)}
-
G^{(i+1),\mathrm{COS}}
\right\|_{\infty,I_{i+1}}.
\label{eq:recursive_payoff_error}
\end{equation}

The positive-part function is $1$-Lipschitz:
\begin{equation}
\left|
(u-K)^+
-
(v-K)^+
\right|
\leq
|u-v|,
\qquad
u,v\in\mathbb R.
\label{eq:positive_part_lipschitz}
\end{equation}

Hence,
\begin{equation}
\begin{aligned}
\left|
G^{(i+1)}(x)
-
G^{(i+1),\mathrm{COS}}(x)
\right|
&=
\left|
\bigl(V^{(i+1)}(x)-K_{i+1}\bigr)^+
-
\bigl(V^{(i+1),\mathrm{COS}}(x)-K_{i+1}\bigr)^+
\right|
\\
&\leq
\left|
V^{(i+1)}(x)
-
V^{(i+1),\mathrm{COS}}(x)
\right|.
\end{aligned}
\label{eq:payoff_error_lipschitz_bound}
\end{equation}

Taking the supremum over $I_{i+1}$ gives
\begin{equation}
\left\|
G^{(i+1)}
-
G^{(i+1),\mathrm{COS}}
\right\|_{\infty,I_{i+1}}
\leq
\varepsilon_{i+1}.
\label{eq:payoff_error_by_continuation_error}
\end{equation}

Combining
\eqref{eq:error_decomposition},
\eqref{eq:first_stagewise_error_term},
\eqref{eq:recursive_payoff_error}, and
\eqref{eq:payoff_error_by_continuation_error}
yields
\begin{equation}
\varepsilon_i
\leq
\delta_i(N_{i+1})
+
\tau_i
+
\Lambda_i\varepsilon_{i+1},
\qquad
i<m.
\label{eq:one_step_recursive_error}
\end{equation}

Applying \eqref{eq:one_step_recursive_error} recursively backward from
$i=m$ gives \eqref{eq:recursive_error_bound}.

Now suppose that \eqref{eq:stagewise_spectral_assumption} holds. Setting $i=0$ in \eqref{eq:recursive_error_bound} gives
\begin{equation}
\varepsilon_0
\leq
\sum_{j=0}^{m}
\left(
\prod_{\ell=0}^{j-1}\Lambda_\ell
\right)
\left[
\delta_j(N_{j+1})
+
\tau_j
\right].
\end{equation}

Since
\[
\delta_j(N_{j+1})
\leq
M_j e^{-\eta_jN_{j+1}}
\leq
M_j e^{-\eta N_{\min}},
\]
and the number of compound stages is fixed, all finite products and sums involving the stability constants are independent of the expansion sizes. Hence, there exist constants $C>0$ and $C_{\mathrm{tr}}>0$ such that
\begin{equation}
\varepsilon_0
\leq
C e^{-\eta N_{\min}}
+
C_{\mathrm{tr}}
\max_{0\leq i\leq m}\tau_i.
\end{equation}

Since $x_0\in I_0$, 
$$
\left|
V^{(0)}(x_0)
-
V^{(0),\mathrm{COS}}(x_0)
\right|
\leq
\varepsilon_0,
$$
which proves \eqref{eq:global_recursive_cos_error}. The final conclusion
follows immediately from \eqref{eq:truncation_cos_balance}.

\end{proof}

\begin{remark}\label{remark:general_payoff}
The recursive error analysis is not restricted to compound call options.
The specific form of the stagewise payoff enters the proof only through the
$1$-Lipschitz property of the positive-part mapping
$
u\mapsto (u-K)^+.
$
More generally, suppose that at stage $i+1$ the payoff can be written as
$$
G^{(i+1)}(x) = \Gamma_{i+1}\!\left(V^{(i+1)}(x)\right),
$$
where $\Gamma_{i+1}:\mathbb R\to\mathbb R$ is Lipschitz continuous with constant $L_{i+1}$. Then, following the same argument in the proof of Theorem \ref{thm:recursive_cos_error}, we can show that 
$$
\varepsilon_i \leq \delta_i(N_{i+1}) + \tau_i + \Lambda_i L_{i+1}\varepsilon_{i+1}.
$$
Hence, provided that the corresponding stagewise COS approximation assumptions hold, the recursive convergence analysis extends directly to more general Lipschitz stagewise payoff transformations. 
\end{remark}

\begin{corollary}[Convergence of the numerical exercise boundaries]
\label{cor:boundary_convergence}

Fix $i\in\{1,\ldots,m\}$, and let $x_i^*$ denote the exact exercise
boundary satisfying
\[
V^{(i)}(x_i^*)=K_i.
\]
Assume that there exists $r_i>0$ such that the closed neighbourhood
\begin{equation}
\mathcal U_i
\coloneqq
[x_i^*-r_i,x_i^*+r_i]
\subset
I_i=[a_i,b_i],
\label{eq:boundary_neighbourhood}
\end{equation}
and that
\begin{equation}
V^{(i)\prime}(x)
\geq
m_i>0,
\qquad
x\in\mathcal U_i.
\label{eq:boundary_derivative_lower_bound}
\end{equation}

If the continuation-value error satisfies
$
\varepsilon_i
\coloneqq
\left\|
V^{(i)}
-
V^{(i),\mathrm{COS}}
\right\|_{\infty,I_i}
<
m_i r_i, 
$
then there exists a numerical exercise boundary
$
\widetilde x_i^{\,*}
\in
(x_i^*-r_i,x_i^*+r_i)
$
satisfying
$$
V^{(i),\mathrm{COS}}(\widetilde x_i^{\,*})=K_i.
$$
Under the single-boundary assumption introduced previously, this numerical boundary is unique in $I_i$. Moreover,
\begin{equation}
\left|
\widetilde x_i^{\,*}-x_i^*
\right|
\leq
\frac{\varepsilon_i}{m_i}
\leq
\frac{1}{m_i}
\sum_{j=i}^{m}
\left(
\prod_{\ell=i}^{j-1}\Lambda_\ell
\right)
\left[
\delta_j(N_{j+1})
+
\tau_j
\right].
\label{eq:global_boundary_error}
\end{equation}

In particular, under the spectral convergence conditions of Theorem~\ref{thm:recursive_cos_error}, if the truncation errors decay at the same spectral rate as the stagewise COS approximation errors, then
\begin{equation}
\left|
\widetilde x_i^{\,*}-x_i^*
\right|
=
\mathcal O\!\left(
e^{-\eta N_{\min}}
\right).
\end{equation}
\end{corollary}

\begin{proof}

Since $V^{(i)}(x_i^*)=K_i$, the mean-value theorem together with
\eqref{eq:boundary_derivative_lower_bound} gives
\begin{equation}
V^{(i)}(x_i^*-r_i)-K_i
\leq
-m_i r_i,
\qquad
V^{(i)}(x_i^*+r_i)-K_i
\geq
m_i r_i.
\end{equation}
The condition $\varepsilon_i<m_i r_i$ ensures that these signs are preserved
by the COS approximation. Indeed,
\begin{equation}
\begin{aligned}
V^{(i),\mathrm{COS}}(x_i^*-r_i)-K_i
\leq
V^{(i)}(x_i^*-r_i)-K_i+\varepsilon_i
\leq
-m_i r_i+\varepsilon_i
<0,
\end{aligned}
\end{equation}
while
\begin{equation}
\begin{aligned}
V^{(i),\mathrm{COS}}(x_i^*+r_i)-K_i
\geq
V^{(i)}(x_i^*+r_i)-K_i-\varepsilon_i
\geq
m_i r_i-\varepsilon_i
>0.
\end{aligned}
\end{equation}

Since $V^{(i),\mathrm{COS}}$ is continuous, the intermediate value theorem implies that there exists at least one
$
\widetilde x_i^{\,*}
\in
(x_i^*-r_i,x_i^*+r_i)
$
such that
$
V^{(i),\mathrm{COS}}(\widetilde x_i^{\,*})=K_i.
$
By the single-boundary assumption, whenever such a solution exists in $I_i$, it is unique. Hence, $\widetilde x_i^{\,*}$ is the unique numerical exercise boundary in $I_i$.

Since both $x_i^*$ and $\widetilde x_i^{\,*}$ lie in $\mathcal U_i$, the mean-value theorem implies that there exists a point $\xi_i$ between them such that
\begin{equation}
\left|
V^{(i)}(\widetilde x_i^{\,*})
-
V^{(i)}(x_i^*)
\right|
=
V^{(i)\prime}(\xi_i)
\left|
\widetilde x_i^{\,*}-x_i^*
\right|.
\label{eq:boundary_mvt}
\end{equation}
Using
$$
V^{(i),\mathrm{COS}}(\widetilde x_i^{\,*})
=
K_i
=
V^{(i)}(x_i^*),
$$
we obtain
\begin{equation}
\begin{aligned}
\left|
V^{(i)}(\widetilde x_i^{\,*})
-
V^{(i)}(x_i^*)
\right|
=
\left|
V^{(i)}(\widetilde x_i^{\,*})
-
V^{(i),\mathrm{COS}}(\widetilde x_i^{\,*})
\right|
\leq
\varepsilon_i.
\end{aligned}
\end{equation}
Since
$
V^{(i)\prime}(\xi_i)\geq m_i,
$
it follows from \eqref{eq:boundary_mvt} that
\[
\left|
\widetilde x_i^{\,*}-x_i^*
\right|
\leq
\frac{\varepsilon_i}{m_i}.
\]
The second inequality in \eqref{eq:global_boundary_error} follows directly from Theorem~\ref{thm:recursive_cos_error}. The spectral convergence statement is then an immediate consequence of the spectral special case established in the same theorem.

\end{proof}

We emphasize that the displacement of the exercise boundary does not introduce a separate contribution to the option-value error. Under the single-boundary condition, integration from the numerical boundary $\widetilde x_i^{\,*}$ is simply an analytic representation of the approximate payoff
$
\bigl(  V^{(i),\mathrm{COS}}-K_i  \bigr)^+.
$
The difference between the exact and approximate payoffs is already controlled by the $1$-Lipschitz property of the positive-part function and is therefore included in the recursive continuation-value error.

The preceding analysis assumes that each numerical boundary equation is solved exactly. If the nonlinear solver is terminated at a nonzero tolerance, the resulting root-finding error constitutes an additional numerical error. The solver tolerance should therefore be chosen sufficiently small relative to the COS approximation error so that it does not affect the overall convergence rate.

\section{Other options within the analytic COS framework}
\label{s3}

In this section, we provide two further examples—chooser options and Bermudan put options—to demonstrate that the analytic COS framework developed in Section~\ref{s2} is not limited to standard compound call options. These examples show that the proposed coefficient construction extends to a broader class of derivatives whose intermediate values can be decomposed into several analytically tractable terms.
Throughout this section, $[a,b]$ denotes the COS truncation interval
for the log-price $X(T_1)$ at the intermediate decision date.

\subsection{Chooser options}

A chooser option grants its holder, at an intermediate decision date \(0<T_1<T\), the right to choose whether the contract becomes a European call or a European put with common maturity \(T\) and strike \(K\). Its value at time \(T_1\) is therefore
\begin{equation}
\begin{aligned}
&\max\left(
V_{\mathrm{call}}(T_1,S(T_1);T,K),
V_{\mathrm{put}}(T_1,S(T_1);T,K)
\right)
\\
=&
V_{\mathrm{call}}(T_1,S(T_1);T,K)
+
\left(
V_{\mathrm{put}}(T_1,S(T_1);T,K)
-
V_{\mathrm{call}}(T_1,S(T_1);T,K)
\right)^+,
\end{aligned}
\end{equation}
where \(V_{\mathrm{call}}(T_1,S(T_1);T,K)\) and
\(V_{\mathrm{put}}(T_1,S(T_1);T,K)\) denote the corresponding European call and put values at time \(T_1\). Hence, the chooser option value decomposes into the value of a European call and an additional non-negative term determined by the difference between the European put and call values.

Let us now formulate the problem within the COS framework. Define the log-price process \(X(t)\coloneqq \log S(t)\) and
\begin{equation}
V_{\mathrm{call}}(x_1)
\coloneqq
V_{\mathrm{call}}(T_1,e^{x_1};T,K),
\qquad
V_{\mathrm{put}}(x_1)
\coloneqq
V_{\mathrm{put}}(T_1,e^{x_1};T,K).
\end{equation}
The payoff at \(T_1\) can then be written as
\begin{equation}\label{chooser_payoff_3}
V_{\mathrm{call}}(x_1)
+
\left(
V_{\mathrm{put}}(x_1) - V_{\mathrm{call}}(x_1)
\right)^+.
\end{equation}

For the first term in \eqref{chooser_payoff_3}, we can apply the COS method for \(V_{\mathrm{call}}(x_1)\) and therefore obtain a trigonometric representation the same as equation \eqref{C_inner_cos_tri}, and accordingly, the outer cosine coefficients of this term have closed-form expressions that are given in Appendix \ref{app:detail_Hout}.

The second term depends on whether the put value exceeds the call value. This induces a switching boundary \(x_1^*\), determined by
$
V_{\mathrm{call}}(x_1^*) = V_{\mathrm{put}}(x_1^*).
$
Under the standard no-arbitrage assumptions, with a non-dividend-paying underlying and constant interest rate \(r\), put--call parity gives
\begin{equation} 
V_{\mathrm{put}}(x_1) - V_{\mathrm{call}}(x_1)
=
K\exp\left(-r(T-T_1)\right)-\exp(x_1).
\end{equation}
Consequently, the switching boundary is explicitly given by
$
x_1^* =\log K-r(T-T_1)
$.

Therefore, the outer cosine coefficients associated with
\(
\left( V_{\mathrm{put}}(x_1) - V_{\mathrm{call}}(x_1) \right)^+
\)
can be written as
\begin{equation}
\frac{2}{b-a}
\int_a^{x_1^*}
\bigl(V_{\mathrm{put}}(x_1)-V_{\mathrm{call}}(x_1)\bigr)
\cos\!\left(\frac{n\pi(x_1-a)}{b-a}\right)\,dx_1.
\end{equation}
Since the integrand is linear in the two continuation values and both admit trigonometric representations, the integral can be evaluated analytically, and the closed-form expressions can be derived similarly to the Appendix \ref{app:detail_Hout}.

The complete outer cosine coefficients are obtained by adding the coefficients of the two terms in \eqref{chooser_payoff_3}. Consequently, the chooser option can be treated within the analytic COS framework developed in Section~\ref{sec:simplecompound}.

\subsection{Bermudan put options}

Let us consider a Bermudan put option with one early exercise date \(0<T_1<T\). The extension to multiple early exercise dates is straightforward. At the early exercise date \(T_1\), the option value is given by
\begin{equation}\label{bermudan_payoff}
\max\left(
\left(K-S(T_1)\right)^+,
V_{\mathrm{put}}(T_1,S(T_1);T,K)
\right),
\end{equation}
and similarly, we can formulate this problem within the COS framework and rewrite the option value at \(T_1\) as
\begin{equation}\label{bermudan_payoff_log}
\max\left(
\left(K-e^{x_1}\right)^+,
V_{\mathrm{put}}(x_1)
\right)
=
V_{\mathrm{put}}(x_1)
+
\left(
\left(K-e^{x_1}\right)^+
-
V_{\mathrm{put}}(x_1)
\right)^+.
\end{equation}

For the first term on the right-hand side of \eqref{bermudan_payoff_log}, we first apply the COS method to \(V_{\mathrm{put}}(x_1)\), which yields a trigonometric representation similar to that in equation \eqref{C_inner_cos_tri}. Consequently, the outer cosine coefficients of this term can be evaluated analytically, and the formulas can be obtained similarly to Section~\ref{sec:simplecompound} and Appendix \ref{app:detail_Hout}.

Under the usual assumptions for a Bermudan put, the exercise boundary
satisfies \(x_1^*<\log K\), and is therefore determined by
\begin{equation} 
K-e^{x_1^*}
=
V_{\mathrm{put}}(x_1^*).
\end{equation}
In contrast to the switching boundary of the chooser option, the Bermudan exercise boundary generally does not admit an explicit expression and must therefore be determined numerically. Under the Black--Scholes framework, or under suitable monotonicity assumptions for more general asset dynamics, this boundary is unique. Immediate exercise is optimal for \( x_1\leq x_1^* \), whereas continuation is optimal for \( x_1> x_1^* \).

Hence,
\begin{equation}\label{bermudan_positive_term}
\left(
\left(K-e^{x_1}\right)^+
-
V_{\mathrm{put}}(x_1)
\right)^+
=
\begin{cases}
K-e^{x_1}-V_{\mathrm{put}}(x_1), & x_1 \leq x_1^*,\\ 
0, & x_1 > x_1^*.
\end{cases}
\end{equation}

Therefore, the outer cosine coefficients associated with the second term in \eqref{bermudan_payoff_log} can be written as
\begin{equation}\label{bermudan_coeff}
\begin{aligned}
& \frac{2}{b-a}
\int_{a}^{x_1^*}
\left(
K - e^{x_1} - V_{\mathrm{put}}(x_1)
\right)
\cos  
\left(
\frac{n\pi(x_1-a)}{b-a}
\right)
\,\mathrm{d}x_1 \\
= &
\frac{2}{b-a}
\int_{a}^{x_1^*}
\left( K-V_{\mathrm{put}}(x_1) \right)
\cos 
\left(
\frac{n\pi(x_1-a)}{b-a}
\right)
\,\mathrm{d}x_1
-
\frac{2}{b-a}
\int_{a}^{x_1^*}
e^{x_1}
\cos 
\left(
\frac{n\pi(x_1-a)}{b-a}
\right)
\,\mathrm{d}x_1.
\end{aligned}
\end{equation}
For the first term on the right-hand side of \eqref{bermudan_coeff}, \(V_{\mathrm{put}}(x)\) is represented by a trigonometric expansion, while \(K\) is constant. The corresponding integral can therefore be evaluated analytically, and the formulas can be obtained similarly to Section~\ref{sec:simplecompound} and Appendix~\ref{app:detail_Hout}. It is also straightforward to check that the second integral also admits a closed-form expression.

The complete outer cosine coefficients are obtained by adding the coefficients of the European continuation value and those in \eqref{bermudan_coeff}, all of which are available in closed form. Consequently, a Bermudan put option with one early exercise date can be treated within the analytic COS framework developed in Section~\ref{sec:simplecompound}. Finally, we can simply repeat the process for Bermudan options with multiple exercise dates,  as all the required coefficients admit closed-form expressions.

\section{Numerical experiments}\label{s4}

This section evaluates the performance of the fully analytic compound COS formulation derived in Section~\ref{s2}. The numerical experiments serve three purposes. We validate the analytic COS recursion by comparing results with known benchmarks under the Black--Scholes GBM dynamics. Moreover, we assess convergence and computational efficiency when replacing the quadrature-based construction of the outer COS coefficients by analytic formulas for the cosine coefficients. Finally, we demonstrate the flexibility of the framework by incorporating non-Gaussian dynamics through characteristic-function-based models, including clustered jump arrivals.

All computations are performed using the analytic compound COS recursion developed in Section~\ref{s2}. Integration intervals are selected using the cumulant-based construction described in Subsection~\ref{sec:interval}. The exercise boundary $x^*$ is obtained from
$
V_{\mathrm{inner}}^{\mathrm{COS}}(x^*)=K_1,
$
and is solved numerically.
In all experiments, the exercise boundary is computed using Brent's root-finding method on the COS integration interval $[a,b]$.

Absolute and relative tolerances are set to $10^{-12}$, so boundary errors are negligible relative to the spectral truncation error.  In practice, convergence is obtained within a few iterations. All reported CPU times include the construction of the COS integration interval and the evaluation of the cumulants required for that construction, as well as the COS recursion. The nonlinear solver overhead is negligible at the stated tolerances.

\subsection{Asset dynamics and characteristic functions}
\label{sec:asset_dynamics}

The compound COS framework requires the conditional characteristic function of the log-value process over each valuation interval. We therefore consider three dynamics for which this characteristic function is available in closed form:
\begin{itemize}
\item geometric Brownian motion (GBM),
\item the Merton jump--diffusion, and
\item clustered jump dynamics driven by the Queue--Hawkes (Q-Hawkes) process of \cite{arias2025heston}.
\end{itemize}
Since GBM and the Merton jump--diffusion arise as simpler limiting cases of the Q-Hawkes specification, we focus on the Q-Hawkes process and briefly discuss the other two models afterwards.

The Q-Hawkes process, introduced into an option-pricing framework in \cite{arias2025heston}, provides an analytically tractable alternative to the classical Hawkes process \cite{Hawkes1971,EmbrechtsHawkes2011}. Both models capture self-excitation: the occurrence of a jump temporarily increases the likelihood of further jumps and thereby generates clustered jump activity. In the classical Hawkes model, the effect of past jumps decays continuously through a deterministic memory kernel, and the
corresponding characteristic function generally does not admit a closed-form expression. In the Q-Hawkes model, the decay of past excitations is instead governed by a stochastic expiration mechanism. This construction preserves the main qualitative features of Hawkes dynamics while yielding a closed-form joint transform of the activation and jump processes.

In a real-option setting, the underlying project value is generally not a traded asset. Valuation therefore need not be carried out under a risk-neutral measure. We model the project value under a physical or project-value measure with diffusion drift $\mu$ and discount expected future cash flows at a project-specific rate $\rho$. Let $S(t)>0$ denote the project value and define its log-value by $X(t)=\log S(t)$. We assume that the log-value evolves according to
\begin{equation}
X(t)
= \log S_0 + \left(\mu-\frac12\sigma^2\right)t + \sigma W(t) + M(t),
\label{eq:qhawkes_logprice_constvol}
\end{equation}
where $W$ is a standard Brownian motion and
\begin{equation}
M(t)
=
\sum_{j=1}^{N(t)}Y_j
\label{eq:qhawkes_jump_component}
\end{equation}
is the cumulative log-jump component. Here, $N(t)$ counts the number of jumps up to time $t$, while $\{Y_j\}_{j\geq1}$ are independent and identically distributed log-jump sizes with characteristic function $\psi_Y(v)$. Thus, each jump time carries a random mark $Y_j$, and a jump of size $Y_j$ changes the project value from $S(t-)$ to $S(t-)e^{Y_j}$.

We use the uncompensated jump specification \eqref{eq:qhawkes_jump_component}. Consequently, $\mu$ represents the drift of the continuous diffusion component, while the jumps contribute separately to the expected growth of the project value. We assume that the Brownian motion $W$, the jump sizes $\{Y_j\}_{j\geq1}$, and the counting processes introduced below are mutually independent.

Jump arrivals are governed by the Queue--Hawkes process of
\cite{arias2025heston}. Its intensity is
\begin{equation}
\lambda(t)
=
\lambda^*+\alpha Q(t),
\label{eq:qhawkes_intensity}
\end{equation}
where $\lambda^*>0$ is the baseline jump intensity, $\alpha>0$ measures the strength of self-excitation, and $Q(t)\in\mathbb{N}_0$ is the activation number. The activation process evolves according to
\begin{equation}
\mathrm{d}Q(t)
=
\mathrm{d}N(t)-\mathrm{d}N^Q(t),
\label{eq:qhawkes_activation}
\end{equation}
where $N$ has stochastic intensity $\lambda(t)$ and $N^Q$ is an expiration process with intensity $\beta Q(t)$, for $\beta>0$.

Each jump of $N$ increases $Q$ by one and therefore raises the jump intensity by $\alpha$, making further jumps temporarily more likely. Conversely, each jump of $N^Q$ decreases $Q$ by one and removes one active excitation. Hence, both the accumulation and the decay of jump activity occur randomly. This differs from the classical Hawkes
process, in which the effect of previous jumps decays continuously through a deterministic memory kernel. We work in the stable regime $\beta>\alpha$.

The principal analytical advantage of the Q-Hawkes specification is that the joint conditional characteristic function of the activation process and the marked jump increment is available in closed form. Because the pair $(Q,M)$ is Markov and $M$ enters through its increment, the conditional distribution over $[s,t]$ depends on the past only through the current activation state $Q(s)$. For $\tau=t-s$ and
$q=Q(s)$, define
\begin{equation}
\psi_{Q,M}(u,v;\tau\mid q)
=
\mathbb{E}\!\left[
\exp\!\left(
iuQ(t)+iv\bigl(M(t)-M(s)\bigr)
\right)
\mid Q(s)=q
\right].
\label{eq:qhawkes_joint_cf_definition}
\end{equation}

Specializing Proposition~6 of \cite{arias2025heston} to the uncompensated jump component \eqref{eq:qhawkes_jump_component}, this joint conditional characteristic function is given explicitly by
\begin{equation}
\begin{aligned}
\psi_{Q,M}(u,v;\tau\mid q)
={}&
\exp\!\left[
\frac{\lambda^*\tau}{2\alpha}
\bigl(\beta-\alpha-f(v)\bigr)
\right]
\\
&\times
\left(
\frac{2f(v)}
{
f(v)+g(u,v)
+
e^{-\tau f(v)}
\bigl(f(v)-g(u,v)\bigr)
}
\right)^{\lambda^*/\alpha}
\\
&\times
\left(
\frac{
\bigl(1-e^{-\tau f(v)}\bigr)
\bigl(2\beta-e^{iu}(\beta+\alpha)\bigr)
+
e^{iu}f(v)
\bigl(1+e^{-\tau f(v)}\bigr)
}
{
f(v)+g(u,v)
+
e^{-\tau f(v)}
\bigl(f(v)-g(u,v)\bigr)
}
\right)^q,
\end{aligned}
\label{eq:qhawkes_joint_cf}
\end{equation}
where
\begin{align}
f(v)
=
\sqrt{
(\beta+\alpha)^2
-
4\alpha\beta\,\psi_Y(v)
}
,\qquad
g(u,v)
=
\beta
+
\alpha\left(
1-2\psi_Y(v)e^{iu}
\right).
\end{align}
Since the argument of the square root is generally complex, we choose the branch that is continuous in $v$ and satisfies $f(0)=\beta-\alpha>0$, where the positivity follows from the stability condition $\beta>\alpha$.

If the payoff depends only on the future project value and not explicitly on the terminal activation state $Q(t)$, the latter can be marginalized out. Setting $u=0$ removes the factor $e^{iuQ(t)}$ while retaining the effect of the Q-Hawkes dynamics on the distribution of the jump increment. Consequently,
\begin{equation}
\begin{aligned}
\psi_M(v;\tau\mid q)
=
\mathbb{E}\!\left[
e^{iv(M(t)-M(s))}
\mid Q(s)=q
\right]
=
\psi_{Q,M}(0,v;\tau\mid q).
\end{aligned}
\label{eq:qhawkes_jump_cf}
\end{equation}

Finally, conditional on $X(s)=x$ and $Q(s)=q$, the independence of the Brownian motion and the Q-Hawkes jump system yields
\begin{equation}
\begin{aligned}
\varphi_X(v;t,s\mid x,q)
:={}&
\mathbb{E}\!\left[
e^{ivX(t)}
\mid X(s)=x,\ Q(s)=q
\right]
\\
={}&
\exp\!\left[
ivx
+
iv\left(\mu-\frac12\sigma^2\right)\tau
-
\frac12\sigma^2v^2\tau
\right]
\psi_{Q,M}(0,v;\tau\mid q).
\end{aligned}
\label{eq:qhawkes_logvalue_cf}
\end{equation}
This is the conditional characteristic function required for the COS
valuation.
Although the Q-Hawkes process does not possess independent increments when considered without the activation process $Q$, this does not affect the numerical examples presented here. In our applications, Q-Hawkes dynamics are used only over a single designated investment stage, with the activation state initialized at the beginning of that stage. Consequently, only a single conditional transition characteristic function is required, which can be incorporated into the COS recursion in exactly the same manner as the increment characteristic functions considered in Section~\ref{s2}. Extending the recursive formulation across multiple Q-Hawkes stages would require retaining the activation process as an additional state variable.

Two classical benchmark models are recovered as special or limiting cases of the Q-Hawkes specification:
\begin{itemize}
\item In the limit $\alpha\downarrow 0$, the activation process no longer affects the jump intensity, which becomes constant, $\lambda(t)=\lambda^*$. The jump-arrival process therefore reduces to a Poisson process, and the model becomes the Merton jump--diffusion. If the log-jump sizes satisfy
$
Y_j\sim\mathcal{N}(\mu_J,\sigma_J^2),
$
then
$$
\psi_Y(v)
=
\exp\!\left(
iv\mu_J-\frac12\sigma_J^2v^2
\right),
$$
and the characteristic function of the jump increment over $[s,t]$ is
\begin{equation}
\psi_M^{\mathrm{Merton}}(v;t,s)
=
\exp\!\left[
\lambda^*(t-s)
\left(
\exp\!\left(
iv\mu_J-\frac12\sigma_J^2v^2
\right)-1
\right)
\right].
\label{eq:merton_jump_cf}
\end{equation}

\item If the jump component is removed, then the model reduces to GBM.
\end{itemize}

Hence, the Q-Hawkes specification provides an analytically tractable framework that encompasses continuous diffusion dynamics, independent Poisson jumps, and self-exciting clustered jump arrivals.

For the construction of the COS truncation interval, we compute the conditional cumulants of the log-value increment from its characteristic function. Let
\begin{equation}
\chi_q(v;\tau)
=
\exp\!\left[
iv\left(\mu-\frac12\sigma^2\right)\tau
-\frac12\sigma^2v^2\tau
\right]
\psi_{Q,M}(0,v;\tau\mid q),
\label{eq:qhawkes_increment_cf}
\end{equation}
so that
\[
\varphi_X(v;t,s\mid x,q)
=
e^{ivx}\chi_q(v;\tau).
\]
The conditional cumulants are then given by
\begin{equation}
c_n(\tau;q)
=
\frac{1}{i^n}
\left.
\frac{\partial^n}{\partial v^n}
\log\chi_q(v;\tau)
\right|_{v=0},
\qquad n=1,2,4.
\label{eq:qhawkes_cumulants}
\end{equation}
For GBM and the Merton jump--diffusion, the dependence on the activation state $q$ disappears. These cumulants are used to determine the integration interval in the COS approximation.

\subsection{Validation under the Black--Scholes model}\label{sec:validation_bs}

We first validate the analytic compound COS formulation under geometric Brownian motion, for which the European call-on-call compound option admits the closed-form solution of Geske \cite{Geske1979}. This provides an exact benchmark for assessing both the pricing accuracy of the proposed method and the computational gain obtained by replacing the numerical quadrature of the outer COS coefficients with their analytic counterparts.

We consider a European call-on-call compound option with parameters
\begin{align*}
S_0 &= 100, \quad
K_1 = 10, \quad
K_2 = 80, \quad
T_1 = 1.0, \quad
T_2 = 2.0, \quad
r = 0.02, \quad
\sigma = 0.40,
\end{align*}
where $K_1$ is the strike price of the compound option, and $K_2$ is the strike price of the underlying European call option.

To examine the accuracy and convergence of the analytic COS formulation, we compute the option price for increasing numbers of cosine terms. Table~\ref{tab:bs_validation_convergence} reports the resulting COS prices together with their absolute errors relative to Geske's closed-form solution.

\begin{table}[ht]
\centering
\begin{tabular}{c c c}
\hline
$N$ & COS Price & Absolute Error \\
\hline
Geske solution & 24.944697282 & -- \\
32  & 20.921940671 & $4.0\times10^{0}$ \\
64  & 24.944708625 & $1.1\times10^{-5}$ \\
128 & 24.944697282 & $<10^{-10}$ \\
256 & 24.944697282 & $<10^{-10}$ \\
\hline
\end{tabular}
\caption{Validation and convergence of the analytic COS formulation
under Black--Scholes dynamics.}
\label{tab:bs_validation_convergence}
\end{table}

The results show rapid convergence toward the closed-form benchmark. With $N=64$, the absolute error is already of order $10^{-5}$, while $N=128$ achieves an error below $10^{-10}$. This confirms the numerical accuracy of the analytic coefficient construction and demonstrates the fast convergence of the COS approximation under GBM dynamics.

We next assess the computational advantage of the analytic coefficient formulation by comparing it with a quadrature-based COS implementation. A reference value is computed using Geske's formula. For the comparison, both methods use $N=64$ cosine terms, while the number of quadrature nodes $n_q$ in the quadrature-based method is increased until an accuracy comparable to that of the analytic formulation is reached.

\begin{table}[ht]
\centering
\begin{tabular}{c c c c}
\hline
Method & Parameters & Error & CPU (s) \\
\hline
Analytic COS & $N=64$ & $1.1\times10^{-5}$ & $3.38\times10^{-4}$ \\
\hline
Quadrature COS & $N=64,\ n_q=160$  & $1.9\times10^{-3}$ & $2.57\times10^{-4}$ \\
Quadrature COS & $N=64,\ n_q=640$  & $2.1\times10^{-5}$ & $5.49\times10^{-4}$ \\
Quadrature COS & $N=64,\ n_q=5120$ & $1.2\times10^{-5}$ & $3.54\times10^{-3}$ \\
\hline
\end{tabular}
\caption{Accuracy and computational cost of the analytic and quadrature-based COS formulations for $K_1=10$. CPU times are averaged over 10 runs.}
\label{tab:cos_timing_comparison}
\end{table}

For a coarse quadrature grid, the quadrature-based formulation can be slightly faster, but this comes at the cost of a substantially larger discretization error. Increasing $n_q$ improves the accuracy but also increases the computational cost. At comparable accuracy, the analytic formulation is faster: for example, with $n_q=5120$, the quadrature-based method attains an error of the same order as the analytic method but requires approximately an order of magnitude more CPU time.

\subsection{Two-stage pharmaceutical R\&D investment under jump risk and clustering}
\label{s3_schyns53_qhawkes}

We next investigate how discontinuous and clustered information arrival affects the valuation of a staged pharmaceutical R\&D investment. During the early development phase, project values may evolve relatively smoothly, reflecting activities such as legal preparation, licensing, and internal development work. We refer to this stage as the \emph{juristic phase}. By contrast, during the subsequent clinical testing phase, project values may respond abruptly to new information, such as trial outcomes, regulatory feedback, or competing developments. Such events can lead to substantial upward or downward revisions of the project value.

To reflect this economic structure, we allow the project-value dynamics to vary across the two stages. The project value follows GBM during the first phase, whereas jump risk is introduced during the clinical phase. Such stage-dependent dynamics are naturally accommodated by the compound COS framework, since the backward recursion requires only the conditional characteristic function associated with each time interval.

We revisit the two-stage R\&D real-options example of \cite[Sec.~5.3]{Schyns2025}. The investment opportunity can be represented as a call-on-call compound option:
\begin{itemize}
\item at time $T_1$, the firm may pay $K_1$ to retain the right to
proceed to the clinical stage;
\item at time $T_2$, it may pay $K_2$ to commercialize the project,
yielding the terminal payoff
\[
(S(T_2)-K_2)^+.
\]
\end{itemize}

Accordingly, the time-$0$ value is
\[
V(x_0)
=
e^{-\rho T_1}
\mathbb{E}\!\left[
\bigl(V_{\mathrm{inner}}(X(T_1))-K_1\bigr)^+
\right],
\]
where $x_0 = \log(S_0)$, $V_{\mathrm{inner}}$ denotes the value at $T_1$ of the commercialization option maturing at $T_2$, and the expectation is taken under the project-value dynamics described in Section~\ref{sec:asset_dynamics}.

The contract parameters are taken from \cite{Schyns2025}:
$$
S_0=150,\qquad
\sigma=0.25,\qquad
T_1=5,\qquad
T_2=9,\qquad
K_1=58.37,\qquad
K_2=197.22.
$$
and the economic parameters are
$
\mu=0.05, 
\rho=0.10.
$
Using the analytic compound COS method developed in Section~\ref{s2}, the GBM benchmark value is reproduced as
$
V_0=15.44.
$

Starting from this benchmark, we introduce jump risk only during the clinical phase, $[T_1,T_2]$, and compare two alternative jump-arrival mechanisms. Under the Merton jump--diffusion, jumps arrive independently according to a Poisson process, whereas under the Q-Hawkes specification, jump arrivals are self-exciting and may cluster over time.

For the Merton model, we set
\[
\lambda=0.30,\qquad
\mu_J=\pm0.25,\qquad
\sigma_J=0.25.
\]
The case $\mu_J=0.25$ represents an upward-biased jump distribution, corresponding to favorable information such as successful clinical outcomes, whereas $\mu_J=-0.25$ represents a downward-biased jump distribution, corresponding to adverse developments.

For the Q-Hawkes specification, we retain the same jump-size distribution and introduce self-excitation through
\[
\lambda^*=0.274,\qquad
\alpha=0.20,\qquad
\beta=2,
\]
with the activation process initialized at the beginning of the clinical phase as
$Q(T_1)=0.$

Thus, the clinical phase starts from the baseline intensity $\lambda^*$, and clustering develops endogenously following subsequent jump arrivals. The parameters satisfy the stability condition $\beta>\alpha$. Moreover, $\lambda^*$ is chosen such that the expected number of jumps over $[T_1,T_2]$ matches that of the Merton model with intensity $\lambda=0.30$. Hence, the two jump specifications share the same jump-size distribution and expected jump count, while differing in the temporal dependence of jump arrivals.

\begin{table}[ht]
\centering
\renewcommand{\arraystretch}{1.2}
\begin{tabular}{l c}
\hline
Model & Compound value \\
\hline
GBM & 15.44 \\
Merton (positive-mean jumps) & 45.57 \\
Merton (negative-mean jumps) & 7.70 \\
Q-Hawkes (clustered positive-mean jumps) & 47.30 \\
Q-Hawkes (clustered negative-mean jumps) & 7.86 \\
\hline
\end{tabular}
\caption{Impact of jump risk and clustered jump arrivals on the value of
the two-stage pharmaceutical R\&D investment.}
\label{tab:schyns53_qhawkes_realistic}
\end{table}

Table~\ref{tab:schyns53_qhawkes_realistic} shows that jump risk during the clinical phase can have a substantial effect on the compound-option value. Upward-biased jumps markedly increase the investment value by introducing the possibility of large favorable revisions to the project value, such as those following successful clinical outcomes. Conversely, downward-biased jumps reduce the option value by incorporating the risk of adverse trial or regulatory developments.

For both jump-size scenarios, the Q-Hawkes specification produces a slightly higher value than the corresponding Merton model. Since the expected number of jumps and the jump-size distribution are matched, this difference is associated with the self-exciting arrival structure. Under Q-Hawkes dynamics, jump activity is more concentrated, producing relatively quiet periods together with episodes of clustered information arrival. For the nonlinear compound-option payoff, our results suggest that this additional dispersion in jump activity can increase the value of optionality. 

This highlights a practical insight: in staged R\&D investments, the timing structure of information arrival may affect valuation as strongly as the magnitude of individual shocks. By accommodating both independent and self-exciting jump arrivals through their characteristic functions, the analytic compound COS framework allows these effects to be compared within the same pricing recursion, providing a flexible tool for assessing model risk in multi-stage real-option problems.

\subsection{Multi-stage pharmaceutical R\&D investment with late-stage jump risk}
\label{s3_multistage_qhawkes}

We conclude with a multi-stage real-options experiment motivated by pharmaceutical R\&D projects. Drug development typically proceeds through a sequence of staged investments, with continuation decisions made as new information becomes available from laboratory studies, clinical trials, and regulatory assessments.

We model the project using four intermediate decision dates,
\[
T_1=1,\qquad
T_2=2,\qquad
T_3=3,\qquad
T_4=4,
\]
followed by a final commercialization decision at
$
T=5.
$

At each decision date $T_j$, the firm may pay the stage cost $K_j$ to retain the right to continue development. If the firm chooses not to continue, the project is abandoned and its continuation value becomes zero. At the terminal date $T$, the firm may pay the commercialization cost $K_{\mathrm{term}}$, yielding the payoff
$$
(S(T)-K_{\mathrm{term}})^+.
$$

The stage costs are
\[
K_1=15,\qquad
K_2=20,\qquad
K_3=30,\qquad
K_4=45,
\]
and the terminal commercialization cost is
$
K_{\mathrm{term}}=190.
$

As in the preceding experiment, valuation is performed under the project-value measure rather than a risk-neutral measure. The project value has drift $\mu$ and volatility $\sigma$, while future project cash flows are discounted at the project-specific rate $\rho$. We use
\[
S_0=150,\qquad
\mu=0.05,\qquad
\rho=0.10,\qquad
\sigma=0.25.
\]
and all stages are computed with $N=1024$ cosine terms.

Early stages of pharmaceutical development often evolve relatively gradually as information accumulates through laboratory work and early clinical studies. By contrast, late-stage clinical or regulatory outcomes may lead to abrupt revisions of the project value. To reflect this distinction, we consider stage-dependent dynamics in which
\begin{itemize}
\item GBM governs the project value over the early stages $[0,T_4]$;
\item jump risk is introduced only during the final stage $[T_4,T]$.
\end{itemize}

We compare three specifications for the final stage: a GBM benchmark, the Merton jump--diffusion with independent jump arrivals, and the Q-Hawkes specification with self-exciting clustered jump arrivals.

For the Merton model, the log-jump sizes follow
$
Y\sim\mathcal{N}(\mu_J,\sigma_J^2),
$
with parameters
\[
\lambda=0.60,\qquad
\mu_J=\pm0.35,\qquad
\sigma_J=0.25.
\]
The positive-mean and negative-mean cases represent favorable and adverse late-stage information, respectively.

For the Q-Hawkes specification, we retain the same jump-size
distribution and use
\[
\lambda^*=0.583,\qquad
\alpha=0.10,\qquad
\beta=2.0,
\]
with the activation process initialized at the beginning of the final stage as
$
Q(T_4)=0.
$
Thus, the final stage begins at the baseline jump intensity $\lambda^*$, after which clustering develops endogenously through self-excitation. The parameters satisfy the stability condition $\beta>\alpha$.

For a direct comparison with the Merton model, $\lambda^*$ is chosen so that the expected number of Q-Hawkes jumps over $[T_4,T]$ matches that of the Merton model. Under both specifications,
$
\mathbb{E}\!\left[N(T)-N(T_4)\right]\approx0.60.
$
Hence, the two jump models share the same jump-size distribution and expected jump count, while differing in the temporal dependence of jump arrivals.

Table~\ref{tab:multistage_pharma} reports the resulting project values.
\begin{table}[ht]
\centering
\begin{tabular}{l c}
\hline
Model specification & Real option value \\
\hline
GBM (all stages) & 0.678 \\
Merton (positive-mean jumps in final stage) & 6.751 \\
Merton (negative-mean jumps in final stage) & 0.135 \\
Q-Hawkes (clustered positive-mean jumps) & 6.937 \\
Q-Hawkes (clustered negative-mean jumps) & 0.138 \\
\hline
\end{tabular}
\caption{Multi-stage pharmaceutical R\&D valuation with jump risk
introduced only during the final stage.}
\label{tab:multistage_pharma}
\end{table}

The GBM benchmark yields a relatively small project value, reflecting the high terminal commercialization cost relative to the initial project value. Introducing favorable late-stage jump risk substantially increases the value of the development opportunity, since successful clinical or regulatory outcomes may lead to large upward revisions of the project value. Conversely, adverse jumps reduce the option value below the diffusion benchmark. Consistent with the preceding two-stage experiment, the Q-Hawkes specification produces slightly higher values than the corresponding Merton model. Since the expected jump count and jump-size distribution are matched, these differences reflect the self-exciting and clustered structure of jump arrivals.

The COS recursion also provides the critical continuation thresholds $x_j^*$ at each decision date, which determine the minimum project value required for continued investment. Favorable late-stage jump opportunities tend to lower the earlier thresholds, since the possibility of a future breakthrough increases the value of remaining in the project. Conversely, adverse late-stage jump risk tends to raise these thresholds, as continuation becomes less attractive in the presence of unfavorable future outcomes.

This experiment highlights an important economic insight: even when discontinuous information arrives only during the final stage of development, its anticipated effect can propagate backwards through the entire sequence of investment decisions and materially affect the value of the project. From a numerical perspective, the analytic multi-stage COS framework accommodates such stage-specific dynamics without altering the recursive valuation structure: only the characteristic function associated with the relevant stage is modified, while the COS recursion and boundary computation remain unchanged.

\section{Conclusion}

This paper develops a fully analytic COS formulation for compound option valuation. Using the trigonometric structure of the COS expansion, the outer cosine coefficients can be evaluated in a closed form, eliminating the need for numerical quadrature while preserving the convergence properties of the method.

The framework extends to multi-stage compound structures, enabling efficient valuation of sequential investment opportunities such as staged R\&D projects. The resulting analytic formulation improves computational efficiency, enhances numerical stability, and remains applicable to a broad class of independent-increment models with known characteristic functions, including stage-dependent jump-diffusion dynamics.

Numerical experiments confirm the accuracy and efficiency of the approach. Under Black--Scholes dynamics, the method reproduces benchmark solutions while delivering substantial computational speed-ups. Applications to real-options settings further demonstrate that incorporating jump risk can materially affect continuation values, and that the analytic COS recursion enables such model variations to be implemented without altering the pricing structure. Overall, the analytic compound COS method provides a robust tool for pricing nested optionality under realistic dynamics, particularly in multi-stage decision environments where both computational efficiency and modeling richness are essential.

\paragraph{Acknowledgment.}

The authors thank Felice Schyns, whose BSc project provided inspiration for this work, and Dr.~Gero Junike for his valuable comments and suggestions. The first author gratefully acknowledges scholarship support from the China Scholarship Council.

\bibliographystyle{alpha}
\bibliography{references}

\appendix
\section{Closed-form formulas for the simple compound case}
\label{app:detail_Hout}

In this appendix, we provide the detailed derivation of the closed-form expressions used to evaluate the outer cosine coefficients introduced in Section~\ref{sec:simplecompound}. Recall that these coefficients can be written in terms of three elementary integrals as
\begin{align*}
H_{n}^{\mathrm{out}}
&=
\frac{2}{b-a}
\sum_{k=0}^{N_{\mathrm{in}}-1}{}'
\Bigg[
A_k
\int_{\widetilde x^*}^{b}
\cos(\omega_k(x-a'))
\cos(\nu_n(x-a))
\,\d x
\\
&\qquad\qquad
-
B_k
\int_{\widetilde x^*}^{b}
\sin(\omega_k(x-a'))
\cos(\nu_n(x-a))
\,\d x
\Bigg]
\\
&\quad
-
\frac{2K_1}{b-a}
\int_{\widetilde x^*}^{b}
\cos(\nu_n(x-a))
\,\d x \\
& \equiv \sum_{k=0}^{N_{\rm{in}}-1}{}'
\left(
A_k I_{k,n}^{(c)}
-
B_k I_{k,n}^{(s)}
\right)
-
K_1 I_n^{(o)}
\end{align*}
where 
\begin{align*}
I_{k,n}^{(c)}
&=
\frac{2}{b-a}
\int_{\widetilde x^*}^{b}
\cos(\omega_k(x-a'))
\cos(\nu_n(x-a))
\,\d x,
\\
I_{k,n}^{(s)}
&=
\frac{2}{b-a}
\int_{\widetilde x^*}^{b}
\sin(\omega_k(x-a'))
\cos(\nu_n(x-a))
\,\d x,
\\
I_n^{(o)}
&=
\frac{2}{b-a}
\int_{\widetilde x^*}^{b}
\cos(\nu_n(x-a))
\,\d x.
\end{align*}
Next, we derive the closed-form expressions for these three integral terms.

If $\nu_n>0$, that is, $n\geq 1$, we obtain 
\begin{equation}\label{app:I0_npos}
I_n^{(o)}
=
\frac{2}{b-a}
\left[
\frac{\sin(\nu_n(x-a))}{\nu_n}
\right]_{x=\widetilde x^*}^{x=b}.
\end{equation}

For $\nu_0=0$ or equivalently $n=0$, we directly integrate and get
\begin{equation}\label{app:I0_n0}
I_0^{(o)}
=
\frac{2}{b-a}
\int_{\widetilde x^*}^{b}1\,\d x
=
\frac{2(b-\widetilde x^*)}{b-a}.
\end{equation}

For the term $I_{k,n}^{(c)}$, we use 
\begin{equation}
\cos\bar\alpha\cos\bar\beta
=
\frac{1}{2}
\left[
\cos(\bar\alpha-\bar\beta)
+
\cos(\bar\alpha+\bar\beta)
\right], 
\qquad  \forall \bar\alpha, \bar\beta \in \mathbb{R}.
\end{equation}

Let $\bar\alpha=\omega_k(x-a')$ and $\bar\beta=\nu_n(x-a)$. For $\omega_k\neq\nu_n$, we obtain
\begin{align}\label{app:Ic_general}
I_{k,n}^{(c)}
&=
\frac{1}{b-a}
\int_{\widetilde x^*}^{b}
\cos
\bigl(
(\omega_k-\nu_n)x-(\omega_k a'-\nu_n a)
\bigr)
\,\d x
\\
& \quad
+
\frac{1}{b-a}
\int_{\widetilde x^*}^{b}
\cos
\bigl(
(\omega_k+\nu_n)x-(\omega_k a'+\nu_n a)
\bigr)
\,\d x
\nonumber\\
&=
\frac{1}{b-a}
\left[
\frac
{
\sin
\bigl(
(\omega_k-\nu_n)x-(\omega_k a'-\nu_n a)
\bigr)
}
{
\omega_k-\nu_n
}
\right.
\\
&\qquad\qquad\left.
+
\frac
{
\sin
\bigl(
(\omega_k+\nu_n)x-(\omega_k a'+\nu_n a)
\bigr)
}
{
\omega_k+\nu_n
}
\right]_{x=\widetilde x^*}^{x=b},
\end{align}

For $\omega_k = \nu_n = \omega>0$, we use
\begin{equation}
\cos(\omega(x-a'))\cos(\omega(x-a))
=
\frac{1}{2}\cos\!\bigl(\omega(a-a')\bigr)
+
\frac{1}{2}\cos\!\bigl(2\omega x-\omega(a+a')\bigr),
\end{equation}
which gives
\begin{equation}\label{app:Ic_resonant}
I_{k,n}^{(c)}
=
\frac{2}{b-a}
\left[
\frac{b-\widetilde x^*}{2}
\cos\!\bigl(\omega(a-a')\bigr)
+
\frac{1}{2}
\left[
\frac{
\sin\!\bigl(2\omega x-\omega(a+a')\bigr)
}{
2\omega
}
\right]_{x=\widetilde x^*}^{x=b}
\right].
\end{equation}

For $\omega_k = \nu_n = \omega = 0$, corresponding to $k=n=0$, the integrand is constant and direct integration yields
\begin{equation}\label{app:Ic_n0}
I_{0,0}^{(c)} = \frac{2(b-\widetilde x^*)}{b-a}.
\end{equation}

Similarly, for the term $I_{k,n}^{(s)}$ we use
\begin{equation}
\sin\bar\alpha\cos\bar\beta
=
\frac{1}{2}
\left[
\sin(\bar\alpha+\bar\beta)
+
\sin(\bar\alpha-\bar\beta)
\right], 
\qquad  \forall \bar\alpha, \bar\beta \in \mathbb{R}.
\end{equation}

With the same choice for $\bar\alpha$ and $\bar\beta$, and for $\omega_k\neq\nu_n$  we obtain
\begin{align}\label{app:Is_general}
I_{k,n}^{(s)}
& =
\frac{1}{b-a}
\int_{\widetilde x^*}^{b}
\sin 
\bigl(
(\omega_k+\nu_n)x-(\omega_k a'+\nu_n a)
\bigr)
\,\d x  \\
&\quad
+
\frac{1}{b-a}
\int_{\widetilde x^*}^{b}
\sin 
\bigl(
(\omega_k-\nu_n)x-(\omega_k a'-\nu_n a)
\bigr)
\,\d x  \\
&=
-\frac{1}{b-a}
\left[
\frac{
\cos
\bigl(
(\omega_k+\nu_n)x-(\omega_k a'+\nu_n a)
\bigr)
}{
\omega_k+\nu_n
}
\right.   \\
&\qquad\qquad\left.
+
\frac{
\cos 
\bigl(
(\omega_k-\nu_n)x-(\omega_k a'-\nu_n a)
\bigr)
}{
\omega_k-\nu_n
}
\right]_{x=\widetilde x^*}^{x=b},
\end{align}

For $\omega_k=\nu_n=\omega>0$, we use
\begin{equation}
\sin(\omega(x-a'))\cos(\omega(x-a))
=
\frac{1}{2}
\sin \bigl(2\omega x-\omega(a+a')\bigr)
+
\frac{1}{2}
\sin \bigl(\omega(a-a')\bigr),
\end{equation}
and hence
\begin{equation}\label{app:Is_resonant}
I_{k,n}^{(s)}
=
\frac{2}{b-a}
\left[
\frac{b-\widetilde x^*}{2}
\sin \bigl(\omega(a-a')\bigr)
-
\frac{1}{2}
\left[
\frac{
\cos \bigl(2\omega x-\omega(a+a')\bigr)
}{
2\omega
}
\right]_{x=\widetilde x^*}^{x=b}
\right].
\end{equation}

For $\omega_k=\nu_n=\omega = 0$, that is, $k=n=0$, the sine factor vanishes identically and therefore
\begin{equation}\label{app:Is_n0}
I_{0,0}^{(s)}=0.
\end{equation}

\noeqref{app:I0_npos}
\noeqref{app:I0_n0}
\noeqref{app:Ic_general} 
\noeqref{app:Ic_resonant}
\noeqref{app:Ic_n0}
\noeqref{app:Is_general}
\noeqref{app:Is_resonant}
\noeqref{app:Is_n0}

Since the coefficients $A_k$ and $B_k$ have already been computed in the construction of $V_{\mathrm{inner}}^{\mathrm{COS}}$, combining them with the closed-form expressions \eqref{app:I0_npos}--\eqref{app:Is_n0} yields an analytic evaluation of $H_{n}^{\mathrm{out}}$. These coefficients can then be inserted directly into the outer COS valuation formula, eliminating the need for numerical quadrature in the evaluation of the outer payoff coefficients.

Moreover, for fixed $k$ and $n$, the integral terms \eqref{app:I0_npos}--\eqref{app:Is_n0} depend only on the truncation intervals $[a,b]$ and $[a',b']$ and the numerical exercise boundary $\widetilde{x}^*$. This structure will be useful when extending the analytic construction from the simple compound option to the multi-stage compounding case.

\section{Closed-form formulas for the multiple compound case}
\label{app:multistage_analytic_coeffs}

The derivation of the COS coefficients for the multi-stage compounding case follows the same structure as in Appendix~\ref{app:detail_Hout}, as a consequence of the invariant trigonometric structure across successive compounding layers.

Recall from Section~\ref{sec:multi_compound} that $V^{(m),\mathrm{COS}}(x_m)$ is obtained directly from the standard COS approximation of the terminal European call payoff and therefore admits an explicit trigonometric representation. Suppose recursively that, for some $i=m-1,\ldots,0$, the COS approximation $V^{(i+1),\mathrm{COS}}$ is available in the form
\begin{equation}
V^{(i+1),\mathrm{COS}}(x_{i+1})
=
\sum_{k=0}^{N_{i+2}-1}{}'
\left(
A_k^{(i+1)}
\cos\!\left(
\omega_k^{(i+2)}(x_{i+1}-a_{i+2})
\right)
-
B_k^{(i+1)}
\sin\!\left(
\omega_k^{(i+2)}(x_{i+1}-a_{i+2})
\right)
\right),
\end{equation}
where $A_k^{(i+1)}$ and $B_k^{(i+1)}$ depend on $H_k^{(i+2)}$ through their definitions and are available from the preceding backward step.

To compute $V^{(i),\mathrm{COS}}(x_i)$, we first determine the payoff coefficients at time $T_{i+1}$. Under the single-boundary condition introduced in Section~\ref{sec:multi_compound}, these coefficients are given by
\begin{equation}
H_n^{(i+1)}
=
\frac{2}{b_{i+1}-a_{i+1}}
\int_{\widetilde{x}_{i+1}^*}^{b_{i+1}}
\left(
V^{(i+1),\mathrm{COS}}(x_{i+1})-K_{i+1}
\right)
\cos\!\left(
\omega_n^{(i+1)}(x_{i+1}-a_{i+1})
\right)
\,\mathrm d x_{i+1}.
\end{equation}

Substituting the trigonometric representation of $V^{(i+1),\mathrm{COS}}(x_{i+1})$ into $H_n^{(i+1)}$, we obtain
\begin{align*}
H_n^{(i+1)}
&=
\frac{2}{b_{i+1}-a_{i+1}}
\sum_{k=0}^{N_{i+2}-1}{}'
\Bigg[
A_k^{(i+1)}
\int_{\widetilde x^*_{i+1}}^{b_{i+1}}
\cos\!\left(
\omega_k^{(i+2)}(x_{i+1}-a_{i+2})
\right)
\cos\!\left(
\omega_n^{(i+1)}(x_{i+1}-a_{i+1})
\right)
\,\mathrm d x_{i+1}
\\
&\qquad\qquad
-
B_k^{(i+1)}
\int_{\widetilde x^*_{i+1}}^{b_{i+1}}
\sin\!\left(
\omega_k^{(i+2)}(x_{i+1}-a_{i+2})
\right)
\cos\!\left(
\omega_n^{(i+1)}(x_{i+1}-a_{i+1})
\right)
\,\mathrm d x_{i+1}
\Bigg]
\\
&\quad
-
\frac{2K_{i+1}}{b_{i+1}-a_{i+1}}
\int_{\widetilde x^*_{i+1}}^{b_{i+1}}
\cos\!\left(
\omega_n^{(i+1)}(x_{i+1}-a_{i+1})
\right)
\,\mathrm d x_{i+1}
\\
&=
\sum_{k=0}^{N_{i+2}-1}{}'
\left(
A_k^{(i+1)} I_{k,n}^{(c,i+1)}
-
B_k^{(i+1)} I_{k,n}^{(s,i+1)}
\right)
-
K_{i+1} I_n^{(o,i+1)},
\end{align*}
where
\begin{align}
I_{k,n}^{(c,i+1)}
&:=
\frac{2}{b_{i+1}-a_{i+1}}
\int_{\widetilde x^*_{i+1}}^{b_{i+1}}
\cos\!\left(
\omega_k^{(i+2)}(x_{i+1}-a_{i+2})
\right)
\cos\!\left(
\omega_n^{(i+1)}(x_{i+1}-a_{i+1})
\right)
\,\mathrm d x_{i+1},
\\
I_{k,n}^{(s,i+1)}
&:=
\frac{2}{b_{i+1}-a_{i+1}}
\int_{\widetilde x^*_{i+1}}^{b_{i+1}}
\sin\!\left(
\omega_k^{(i+2)}(x_{i+1}-a_{i+2})
\right)
\cos\!\left(
\omega_n^{(i+1)}(x_{i+1}-a_{i+1})
\right)
\,\mathrm d x_{i+1},
\\
I_n^{(o,i+1)}
&:=
\frac{2}{b_{i+1}-a_{i+1}}
\int_{\widetilde x^*_{i+1}}^{b_{i+1}}
\cos\!\left(
\omega_n^{(i+1)}(x_{i+1}-a_{i+1})
\right)
\,\mathrm d x_{i+1}.
\end{align}

These integrals have the same structure as those derived in Appendix~\ref{app:detail_Hout}. Applying the same trigonometric identities yields the following closed-form expressions.

For $I_n^{(o,i+1)}$, if $\omega_n^{(i+1)}>0$, we have
\begin{equation}
I_n^{(o,i+1)}
=
\frac{2}{b_{i+1}-a_{i+1}}
\left[
\frac{
\sin\!\left(
\omega_n^{(i+1)}(x_{i+1}-a_{i+1})
\right)
}{
\omega_n^{(i+1)}
}
\right]_{x_{i+1}=\widetilde x^*_{i+1}}^{x_{i+1}=b_{i+1}}.
\end{equation}
If $\omega_n^{(i+1)}=0$, then
\begin{equation}
I_n^{(o,i+1)}
=
\frac{2(b_{i+1}-\widetilde x^*_{i+1})}
{b_{i+1}-a_{i+1}}.
\end{equation}

For $I_{k,n}^{(c,i+1)}$, if
$\omega_k^{(i+2)}\neq\omega_n^{(i+1)}$, then
\begin{align}
I_{k,n}^{(c,i+1)}
&=
\frac{1}{b_{i+1}-a_{i+1}}
\left[
\frac{
\sin\!\left(
(\omega_k^{(i+2)}-\omega_n^{(i+1)})x_{i+1}
-
(\omega_k^{(i+2)}a_{i+2}
-\omega_n^{(i+1)}a_{i+1})
\right)
}{
\omega_k^{(i+2)}-\omega_n^{(i+1)}
}
\right.
\\
&\qquad\qquad\left.
+
\frac{
\sin\!\left(
(\omega_k^{(i+2)}+\omega_n^{(i+1)})x_{i+1}
-
(\omega_k^{(i+2)}a_{i+2}
+\omega_n^{(i+1)}a_{i+1})
\right)
}{
\omega_k^{(i+2)}+\omega_n^{(i+1)}
}
\right]_{x_{i+1}=\widetilde x^*_{i+1}}^{x_{i+1}=b_{i+1}}.
\end{align}

If
$\omega_k^{(i+2)}=\omega_n^{(i+1)}=\omega>0$, then
\begin{equation}
I_{k,n}^{(c,i+1)}
=
\frac{2}{b_{i+1}-a_{i+1}}
\left[
\frac{b_{i+1}-\widetilde x^*_{i+1}}{2}
\cos\!\left(
\omega(a_{i+1}-a_{i+2})
\right)
+
\frac{1}{2}
\left[
\frac{
\sin\!\left(
2\omega x_{i+1}-\omega(a_{i+1}+a_{i+2})
\right)
}{
2\omega
}
\right]_{x_{i+1}=\widetilde x^*_{i+1}}^{x_{i+1}=b_{i+1}}
\right].
\end{equation}

If
$\omega_k^{(i+2)}=\omega_n^{(i+1)}=\omega=0$, then
\begin{equation}
I_{k,n}^{(c,i+1)}
=
\frac{2(b_{i+1}-\widetilde x^*_{i+1})}
{b_{i+1}-a_{i+1}}.
\end{equation}

For $I_{k,n}^{(s,i+1)}$, if
$\omega_k^{(i+2)}\neq\omega_n^{(i+1)}$, we obtain
\begin{align}
I_{k,n}^{(s,i+1)}
&=
-\frac{1}{b_{i+1}-a_{i+1}}
\left[
\frac{
\cos\!\left(
(\omega_k^{(i+2)}+\omega_n^{(i+1)})x_{i+1}
-
(\omega_k^{(i+2)}a_{i+2}
+\omega_n^{(i+1)}a_{i+1})
\right)
}{
\omega_k^{(i+2)}+\omega_n^{(i+1)}
}
\right.
\\
&\qquad\qquad\left.
+
\frac{
\cos\!\left(
(\omega_k^{(i+2)}-\omega_n^{(i+1)})x_{i+1}
-
(\omega_k^{(i+2)}a_{i+2}
-\omega_n^{(i+1)}a_{i+1})
\right)
}{
\omega_k^{(i+2)}-\omega_n^{(i+1)}
}
\right]_{x_{i+1}=\widetilde x^*_{i+1}}^{x_{i+1}=b_{i+1}}.
\end{align}

If
$\omega_k^{(i+2)}=\omega_n^{(i+1)}=\omega>0$, then
\begin{equation}
I_{k,n}^{(s,i+1)}
=
\frac{2}{b_{i+1}-a_{i+1}}
\left[
\frac{b_{i+1}-\widetilde x^*_{i+1}}{2}
\sin\!\left(
\omega(a_{i+1}-a_{i+2})
\right)
-
\frac{1}{2}
\left[
\frac{
\cos\!\left(
2\omega x_{i+1}-\omega(a_{i+1}+a_{i+2})
\right)
}{
2\omega
}
\right]_{x_{i+1}=\widetilde x^*_{i+1}}^{x_{i+1}=b_{i+1}}
\right].
\end{equation}

If
$\omega_k^{(i+2)}=\omega_n^{(i+1)}=\omega=0$, then
\begin{equation}
I_{k,n}^{(s,i+1)}=0.
\end{equation}

Using these closed-form expressions, all payoff coefficients $H_n^{(i+1)}$, $n=0,\ldots,N_{i+1}-1$, can be evaluated analytically. The coefficients $A_n^{(i)}$ and $B_n^{(i)}$ then follow directly from \eqref{multi_AB_coeff}, yielding the trigonometric representation of $V^{(i),\mathrm{COS}}(x_i)$. Starting from the explicitly available $V^{(m),\mathrm{COS}}$, the same procedure can be applied recursively for $i=m-1,\ldots,0$, ultimately yielding $V^{(0),\mathrm{COS}}(x_0)$ without introducing numerical quadrature at any intermediate compounding stage.

\end{document}